\documentclass[a4paper,USenglish]{lipics-v2021} 

\graphicspath{{./graphics/}}

\title{Less is More: Fairness in Wide-Area Proof-of-Work Blockchain Networks} 

\author{Yifan Mao}{the Ohio State University}{}{}{}

\author{Shaileshh Bojja Venkatakrishnan}{the Ohio State University}{}{}{}

\authorrunning{J. Open Access and J.\,R. Public} 

\Copyright{Jane Open Access and Joan R. Public} 



\category{} 

\relatedversion{} 

\acknowledgements{I want to thank \dots}

\nolinenumbers 

\EventEditors{John Q. Open and Joan R. Access}
\EventNoEds{2}
\EventLongTitle{42nd Conference on Very Important Topics (CVIT 2016)}
\EventShortTitle{CVIT 2016}
\EventAcronym{CVIT}
\EventYear{2016}
\EventDate{December 24--27, 2016}
\EventLocation{Little Whinging, United Kingdom}
\EventLogo{}
\SeriesVolume{42}
\ArticleNo{23}

\newtheorem{lem}{Lemma}
\newtheorem{thm}{Theorem}

\begin{document}

\maketitle

\begin{abstract}
Blockchains are becoming increasingly important in today's Internet, enabling large-scale decentralized applications with strong security and transparency properties. 
In a blockchain system, participants record and update the server-side state of an application as blocks of a replicated, immutable ledger using a consensus protocol over a peer-to-peer network. 
Mining blocks has become lucrative in recent years; e.g., a miner receives over USD 200,000 per mined block in Bitcoin today.  
A key factor affecting mining rewards is the latency of broadcasting blocks over the network. 
In this paper, we consider the problem of topology design for optimizing mining rewards in a wide-area blockchain network that uses a Proof-of-Work protocol for consensus. 
Contrary to general wisdom that a faster network is always better for a miner, we show that increasing network connectivity (e.g., by adding more neighbors) is beneficial to a miner only up to a point after which the miner's rewards degrade.   
This is because when a miner improves its connectivity, it inadvertently also aids other miners in increasing their connectivity.  
An optimal action for a miner, then, is to increase connectivity only to an extent that the miner's own gains strictly outweigh the gains of other miners. 
Similarly, for a set of miners colluding to increase connectivity between themselves, we show that there exists an optimal size for the set at which rewards are maximized. 
A collusion set larger than this optimal diminishes the reward per miner within the set.  
We formalize the topology design problem, and provide extensive experimental results and theoretical evidence to support our claim.   
\end{abstract}

\section{Introduction}
\label{s: intro}

Blockchains have emerged as a secure, transparent and decentralized alternative to centralized servers for implementing important applications such as payments, digital marketplace and smart contracts~\cite{abou2019blockchain}.
Blockchains function over a peer-to-peer (p2p) overlay, where each node in the overlay maintains a replica of an immutable ledger called the blockchain, in which transactions (e.g., payments) are recorded as they are made over the network.  
To function correctly, it is essential for the blockchain replicas at different nodes to be in agreement with each other which is achieved using a consensus protocol. 
Since the Nakamoto consensus~\cite{genesis} was proposed in Bitcoin in 2008, several alternative consensus algorithms have been proposed in recent years with useful performance properties~\cite{Gilad2017algorand, bagaria2019prism}. 
Nevertheless, a number of mainstream blockchains---notably cryptocurrencies---rely on the Nakamoto consensus, also called as proof-of-work (PoW), for their operation~\cite{genesis,ethereum,Dogecoin,mukhopadhyay2016brief}.  
The combined market capitalization of these cryptocurrencies currently stands over a trillion USD~\cite{cryptmarkcap}. 

In a PoW cryptocurrencies, new payment transactions are consolidated into blocks and appended to the blockchain in a process called mining. 
A block in Bitcoin, for instance, contains an average of around 2000 payment transactions~\cite{avgtxnpblk}. 
New blocks are generated by nodes in the p2p overlay, called miners, that compete to solve a computationally difficult cryptographic hash puzzle for each new block. 
The computational challenge involved in creating new blocks dissuades attackers from readily generating incorrect blocks and secures the blockchain. 
Blocks are mined in sequence with each mined block holding a reference to a unique parent block that was mined before the block.
When a miner mines a new block, it broadcasts the block over the p2p network so that other miners can start mining the next block using the received block as a parent. 

Miners are rewarded in proportion to the amount of computational effort they put in for mining new blocks~\cite{kroll2013economics}. 
The amount of rewards earned by a miner depends on the number of blocks the miner mined that has been included in the blockchain. 
Mining is an expensive undertaking---miners often invest in specialized computational hardware and pay hefty electricity fees to keep their infrastructure running~\cite{thum2018economic}. 
Thus miners are incentivized to mine as many valid blocks as possible to maximize their individual payoffs and offset investment cost.  
Increasing the computational effort (i.e., the hashrate of a miner~\cite{rosenfeld2014analysis}) is a natural approach to increasing the number of blocks mined, and hence the rewards earned.  
For a fixed hashrate a miner's reward crucially depends on how well-connected the miner is to other miners in the p2p network; connectivity affects a miner's   
ability to send and receive blocks faster than the other miners over the network~\cite{wan2019evaluating,decker2013information,gencer2018decentralization}. 
Recent works have explored various ways including topology redesigns~\cite{mao2020perigee,rohrer2019kadcast,park2019nodes}, fast relay networks~\cite{falcon,klarman2018bloxroute,fibre}, and coding or compressing blocks~\cite{ozisik2017graphene,chawla2019velocity} to reduce the latency of block propagation in the network.
The conventional wisdom here is that a faster network that is able to disseminate blocks quickly is generally more desirable than a slower network, as it allows miners to hear about newly mined blocks quicker thus reducing the likelihood of mining new blocks over older blocks (a process called as forking~\S\ref{s: background}).
A faster network also facilitates a higher rate of mining blocks, which has positive implications for increasing the rate at which payment transactions are confirmed in the blockchain (i.e., the throughput) and reducing transaction confirmation delays~\cite{bagaria2019prism}. 

In this paper, we study the impact of the network topology on the mining reward of each miner in PoW blockchains operating over wide area p2p networks such as Bitcoin or Ethereum.  
Miners in these networks are often geographically far away from each other resulting in significant (e.g., 100s to 1000s of milliseconds or more~\cite{decker2013information}) propagation delays while broadcasting blocks~\cite{croman2016scaling,decker2013information} compared to the rate of block generation in the network. 
For instance, blocks are mined once every 13 seconds on average in Ethereum, once every minute in Dogecoin, once every 2.5 minutes in Litecoin and once every 10 minutes in Bitcoin. 
In private blockchain networks, the mining rate can be configured to be even faster~\cite{schaffer2019performance,hari2019accel}.  

In these cases, the topology of the p2p overlay has a significant effect on the dynamics of block propagation and consequently the rewards earned by competing miners. 
We consider topology design as a game played between miners, where the number and choice of peers (miners) a miner connects with is the action and the average fraction of blocks mined by the miner that are included in the blockchain is the reward for the miner.   
We assume all miners honestly participate in the block mining and dissemination process. 
Since analyzing the space of all possible strategies that miners can follow and their associated rewards is intractable, we study two simple, canonical settings: 
(i) a non-cooperative model where each miner selfishly decides the number and choice of neighbors to maximize its own reward; 
(ii) a cooperative model in which miners collude to form clusters, with each miner choosing neighbors primarily within its own cluster.  

Our main result is a counter-intuitive phenomenon showing the existence of an `optimal' level of connectivity, that is not necessarily the `maximum' possible connectivity, where payoffs are maximized for a miner.
For instance, it is strictly suboptimal for a miner to establish connections to all other miners in a sparse network compared to connecting with a small number of carefully chosen peers.  
In the non-cooperative model, we show that increasing the number of neighbors benefits a miner only up to a point.
Attempting to increase connectivity beyond this point costs the miner due to increased networking resources (e.g., bandwidth) while reducing the reward gains. 
In the cooperative model, we show that miners part of a tightly-connected `dominant' cluster receive higher than their fair share of rewards, while miners not part of the dominant cluster receive suboptimal payoffs. 
Here, the size of the dominant cluster crucially affects the reward gains earned by colluding miners, with the existence of an optimal cluster size where the rewards are maximized. 

The reason we observe reward gains that are non-monotonic with connectivity is because for a miner to maximize its reward, it is not only essential for the miner to be well-connected to the rest of the network (to receive and send blocks fast), but also for the rest of the miners to {\em not} be well-connected so that they receive blocks slower than the miner on average.
A miner that increases its connectivity to the rest of the network by increasing its number of neighbors is not only creating new low-latency paths between itself and others, but also indirectly creating new low-latency paths between other miners through itself.
On the other hand, a miner that has too few neighbors risks having paths that are of higher latency to other miners.   
The right balance occurs by choosing an appropriate number and choice of neighbors to connect with. 
While prior works have studied how a miner's reward increases with increasing its connectivity to the rest of the network~\cite{shahsavari2019theoretical,xiao2020modeling,cao2021characterizing,gencer2018decentralization,putri2018effect,eyal2014majority}, the counter-intuitive phenomenon of rewards decreasing with increasing connectivity has not been studied to our best knowledge.

Experimentally we observe that clustering of miners occurs even with a random topology simply due to the non-uniform geographical distribution of miners around the world.  
For example, if 70\% of all miners in the world are located in either Europe or North America 
(services monitoring cryptocurrency networks report this to be the case~\cite{bitnodescom}), then even with a constant degree random topology the miners in Europe and North America will end up being a dominant cluster with low intra-cluster latency compared to miners in other continents. 
As such, when considering real-world geographies, we observe an unfair advantage to miners of regions where there is a rich concentration of miners and vice versa. 
However, here too we observe that the payoffs to miners in central locations (Europe or North America, in our experiments) degrades if they increase their number of neighbors beyond a threshold. 
Miners in remote locations strictly benefit by increasing the number of connections they make. 

We summarize the key contributions of the paper below. 
\begin{enumerate}
\item We propose a simple theoretical model to analyze the fraction of successful blocks mined by each miner in a wide-area PoW blockchain network with heterogeneous link latencies. 
\item For the noncooperative model, we show that increasing number of neighbors helps only up to a point beyond which rewards do not increase. 
\item For the cooperative model, we derive the payoff each miner receives as a function of size of the dominant cluster. 
We show that the payoffs to miners in the dominant cluster are maximized when the size of the cluster is 70\% of the network size under assumptions on network latency. 
\item We develop an event driven simulator on Omnet++~\cite{omnetppsim} to simulate mining and block propagation over a world-wide network of 270 miners based on Bitcoin's network protocol and validate our theoretical findings. We also present candidate strategies that miners can follow in order to alleviate clustering bias and maximize their rewards. 
\end{enumerate}

\section{Background}
\label{s: background} 

We use Bitcoin as a representative example of a blockchain system, to provide a background on PoW blockchain networks. 
Other PoW blockchains follow a conceptually similar architecture, potentially with different algorithmic parameters. 
Bitcoin is a digital payment system, that allows end-users to send and receive currency between each other in a secure way. 
Payments made over Bitcoin are recorded in an immutable ledger called the blockchain, as an ordered sequence (earlier payments appear before later payments), which allows the blockchain to verify the validity of payments and prevent double spending.  
All transactions made since the inception of Bitcoin in 2008 are registered in the blockchain, making it a fairly large database (over 300 Gigabytes today). 

\smallskip
\noindent
{\em The blockchain p2p network.}
Bitcoin operates over a peer-to-peer network of servers and clients. 
Clients are end-devices running a lightweight `wallet' application for making payment transactions.   
Servers are typically equipped with high-performance computational hardware (ASICs, GPUs etc.) which they use to validate transactions and extend the blockchain.   
Clients connect to one or more servers via TCP connections; servers accept incoming connection requests from clients. 
In addition servers also connect to other servers via TCP. 
Bitcoin specifies each server can have a maximum of 125 connections to other clients or servers. 

\smallskip 
\noindent 
{\em Proof-of-Work consensus.}
The blockchain is structured as an ordered sequence of blocks, where each blocks contains a small number (e.g., 2000) of ordered transactions. 
The very first block in Bitcoin is called as the genesis block. 
Servers in Bitcoin each maintain a local replica of the blockchain, which they use for verifying new payment transactions and extending the blockchain. 
To provide a secure, trustworthy payment service, it is essential for the blockchain replicas at different servers to be consistent with each other.   
Consistency across the blockchain replicas in Bitcoin is achieved via the Proof-of-Work (PoW) consensus protocol~\cite{genesis}. 

When a client makes a payment in Bitcoin, it encodes the payment information as a transaction message and broadcasts it over the p2p network. 
Each server receives transaction messages, verifies their validity (against older payment transactions recorded in the blockchain), batches the verified transactions in to a new block and digitally signs the block . 
For a block to be considered valid a server must also include the solution to a computationally difficult cryptographic hash puzzle with the block. 
Servers solve this puzzle via brute-force guesswork; the time it takes for a server to solve a puzzle is commonly modeled as an exponential random variable whose rate is proportional to the number of hash computations per second the server's computational hardware can perform~\cite{bagaria2019prism}. 
Each block also includes a reference to a unique block on the blockchain called the parent block (a block is a child block of its parent). 
The number of blocks encountered while traversing along the parent links starting from a block, until the genesis block is reached, is called as the height of a block. 
Bitcoin follows the `longest chain protocol' for choosing parent blocks when creating a new block: a new block must point to a block whose height is the greatest on the blockchain.
This process of generating a new block is called as mining. 

When a server mines a block, it immediately broadcasts the block to other servers via flooding over the p2p network. 
Servers receiving a block verify the validity of the block (e.g., validity of transactions in the block, correctness of the solution to the cryptographic puzzle etc.).
If found valid, they append the block in to their local blockchain replica and relay the block to their neighboring peers in the network.
They then start mining the next block using the received block as parent provided it has the greatest height on the blockchain.  
Occasionally, it is possible for two blocks to be concurrently mined by different miners over the same parent block.
This creates a {\em fork} in the blockchain, where there are multiple blocks at the same height.  
In this case a server mining for a new block must mine on top of the block it {\em received the earliest}, from among the blocks at the same (greatest) height. 
Over time, one of the forked chains will eventually grow longer than others to become part of the longest chain. 
Only transactions included in blocks on the longest chain are considered confirmed by the network. 
Blocks not on the longest chain do not contribute to extending the ledger of confirmed transactions. 

\smallskip
\noindent
{\em Mining rewards.} 
To incentivize servers to mine new blocks, a server receives a monetary reward for each block it mines that becomes part of the longest chain. 
The reward comes from two sources. 
First, whenever a payment transaction is made over the blockchain the payer includes a transaction fee in the transaction which goes to the server that verified the transaction and included it in its block as a reward.   
In 2021 the average fee per transaction over Bitcoin was around 3 USD~\cite{btcavgfee}.  

The second source of reward for a server comes from minting new cryptocurrency each time a block is mined. 
In Bitcoin, currently a miner is allowed to mint 6.25 bitcoins (worth over $200,000$ USD today) every time it mines a block.
These newly minted bitcoins are included as a separate transaction in the block, where the payee is the server that mined the block and the payer is a special account address reserved specifically for this purpose. 
Today, the rewards earned from minting new bitcoins represent $> 98 \%$ of the total reward earned by a miner in a block~\cite{txnfeepbr} (compared to reward from transaction fees which is $< 2\%$ of the overall reward). 
However, the amount of bitcoins that a miner can mint per block is set to exponentially diminish over the years in Bitcoin. 
Therefore, transaction fees would eventually become the sole source of reward for miners.  

\smallskip
\noindent 
{\em Impact of network on rewards.} 
Mining is a lucrative business in cryptocurrencies like Bitcoin. 
To maximize the likelihood of a newly mined block becoming part of the longest chain, it is essential that the block is broadcast to other servers as quickly as possible. 
If broadcast is delayed, there is a risk that the block will lose out to a competing block mined at the same height which reached other servers earlier. 
Conversely, a server must also hear about newly mined blocks as quickly as possible. 
This is so that a server does not waste time and computation mining a block at a certain height, when a block at that height has already been mined by another server. 
We use the terms servers, miners, nodes and peers interchangeably in this paper.
We also restrict our focus to only the overlay network of miners, and do not consider client connections as they do not have a significant impact on the block propagation dynamics.

\section{System Model}
\label{s: sysmodel}

\subsection{Network Model}
\label{s:netmodel}

We model the blockchain network as an undirected graph $G(V, E)$ comprising of a set of miners $V$ and communication links (e.g., TCP links as in Bitcoin) $E$ between them. 
Each miner $v \in V$ has an associated hash power $h_v \geq 0$. 
The time it takes for a miner to mine a block on top of a previous block is an exponential random variable with rate $h_v$ (mean $1/h_v$). 
This time it takes for a miner to mine a block, is independent of the time taken to mine prior blocks. 
We consider a uniform hash power, i.e., $h_v = h_u, \forall u,v \in V$ in this paper; however, our results can be extended to non-uniform hash power assignment in a straightforward way.

When a miner mines a block, it broadcasts the block to all of its neighbors in $G$. 
Each link $(u, v)$ has a latency $l_{(u,v)} \geq 0$ which is the time it takes for the block to propagate from $u$ to $v$ (we assume $l_{(u,v)} = l_{(v,u)}$ for any edge $(u,v) \in E$) through link $(u, v)$. 
When a miner $v$ receives a block (that has not been previously received) from neighbor $u$, it spends a time $c_v \geq 0$ in validating the block following which it forwards the block to all of its neighbors other than $u$. 
If a miner receives a block that it has previously seen, it does not forward the block to its neighbors. 

Each miner maintains a local replica of the blockchain based on the blocks it has received so far. 
Before the mining process starts, we assume there is a unique genesis block that is known to all miners. 
For any block $b$, we call the number of blocks on a path from the genesis block to block $b$  on the blockchain as the height of block $b$.
The genesis block has a height of 0.   
Following the longest chain protocol, a miner mines a new block on top of the block with the greatest height in its local blockchain. 
If there are multiple longest chains, then mining is done on top of the earliest received block from among the blocks having the greatest height.
Before a miner could mine a block, if it receives a new block that extends its longest chain height, the miner restarts its mining to mine on top of the newly received block. 
If a miner successfully mines a new block, it immediately starts mining the next block on top of the new block. 
We assume all miners follow protocol, and leave a discussion on various possible adversarial actions that miners can take to future work.  
We do not model transaction generation and broadcast. 

\subsection{Reward Model}

When a block mined by a miner is included in the longest chain, the miner incurs a reward for mining the block as explained in Section~\ref{s: background}. 
While a number of factors affect the rewards earned by a miner per block---such as the plaintext size of the transactions in the block, number of unconfirmed transactions currently in the network, height of block etc.---for simplicity we assume the reward per block included in the longest chain to be a constant. 
Using this model as a motivation, we consider a reward model in which the aggregate reward earned by a miner is proportional to the fraction of the blocks mined by the miner in the longest chain. 
For a miner $v \in V$, we denote by $F_v \in [0, 1]$ the average fraction of blocks mined by $v$ that are included in the longest chain over a long time horizon, i.e., assuming the following limit exists
\begin{align}
F_v = \lim_{T \rightarrow \infty } \frac{\sum_{b \in \mathcal{B}(T)} \mathbf{1}_{\text{block } b \text{ is mined by } v \text{ and is included in the longest chain}}}{\sum_{b \in \mathcal{B}(T)} \mathbf{1}_{\text{block } b \text{ is included in the longest chain}}}, \label{eq: Fv define cont}
\end{align}
where $\mathcal{B}(t)$ denotes the set of all blocks mined up to time $t \geq 0$ and $\mathbf{1}_{(\cdot)}$ is the indicator function. 

In a fair network with $n$ miners, $F_v = 1/n$ for all $v \in V$. 
A network is unfair if there exists miners that have a greater than $1/n$ fraction of blocks mined by them in the longest chain, and others that have less than $1/n$ fraction of blocks in the longest chain. 
We present how miners can gain an advantage by carefully choosing whom to connect with in the network, and what the affected miners can do to alleviate their disadvantage.\footnote{Since $\sum_{v\in V} F_v = 1$, if there exist miners that receive greater than their fair share of rewards, there must exist miners that receiver lesser than their fair share.} 
Outside of choosing neighbors carefully, the miners do not deviate from protocol in any way. 
Choosing neighbors carefully in a blockchain network (but otherwise following protocol) is not considered as adversarial behavior. 
Hence, the policies we present in this paper are ``legal'' methods by which miners can obtain an unfair advantage over others.  

We also model the costs incurred by miners for electricity, cooling etc. while mining blocks.  
If the average cost incurred per mined block (whether or not it is included in the longest chain) is not negligibly small compared to rewards obtained per block, then a miner is not only interested in maximizing the fraction of blocks it mined that become part of the longest chain, but also in minimizing the fraction of wasted blocks it mined that did not become part of the longest chain. 
Note that maximizing the fraction of blocks mined by a miner included in the longest chain is distinct from minimizing the fraction of blocks mined by a miner that have been excluded from the longest chain. 
For a miner $v \in V$, we let $W_v$ be the fraction of blocks mined by $v$ that are excluded from the longest chain over a long time horizon, i.e., assuming the following limit exists 
\begin{align}
W_v = \lim_{T \rightarrow \infty} \frac{\sum_{b \in \mathcal{B}(T)} \mathbf{1}_{\text{block } b \text{ is mined by } v \text{ and is excluded from the longest chain}}}{\sum_{b\in \mathcal{B}(T)} \mathbf{1}_{\text{block } b \text{ is mined by } v}}.
\end{align}
Depending on the precise values for reward earned by a miner per block in the longest chain, and the cost incurred per mined block, a miner $v$ may be interested in maximizing a weighted difference of its $F_v$ and $W_v$ values. 
We treat $F_v$ and $W_v$ as separate performance metrics in the topologies that we consider.  

\subsection{Action Model}

Blockchain networks use a decentralized algorithm for constructing the p2p network. 
We view topology construction as a game between the $n$ miners, and study simple classes of topologies that the miners can construct to maximize their rewards and/or minimize wasted cost.
We assume each miner $v \in V$ is located at a fixed geographical location within a wide area (e.g., the world or a large continent).
The location of the miners induces an underlying latency $l_{(u,v)}$ between every pair of miners $u$ and $v$.
If $u$ and $v$ have a link between them, this is the latency of sending a block from $u$ to $v$ (or vice-versa) over this link. 
If $u$ and $v$ do not have a link between them, this is the latency with which a block sent from $u$ to $v$ (or vice-versa) would propagate had there been a link between $u$ and $v$. 
We assume a miner $u$ has knowledge about $l_{(u,v)}$ for all $v \in V$, but does not know the latencies between other miner pairs. 
We provide two control knobs that miners can use to change topology: 
\begin{enumerate}
\item {\em Degree}. A miner $v$ specifies a degree $d_v$ which allows $v$ to make $d_v$ outgoing connections; 
\item {\em Choice of neighbors}. A miner $v$ specifies the set of $d_v$ miners to which to connect to. 
\end{enumerate}
We assume each miner knows the IP address of all other miners in the network. 

We consider two settings, noncooperative and cooperative, by which miners can change topology to their advantage. 
In the noncooperative model, miners do not collude.
Each miner chooses its neighbors on its own without any assistance from other nodes. 
In the cooperative model, miners are allowed to collude and choose neighbors preferentially from the miners they are colluding with.

\section{Motivation}
\label{s:motivation} 

We motivate our results using a simple toy example, that illustrates the effect of network topology on mining rewards. 
We consider a network $G'(V', E')$ with 99 nodes connected as a random graph of average degree 4, with each node initiating connections to 2 random neighbors. 
Let $v$ be a 100-th node that is looking to maximize its rewards by choosing its neighbors carefully among the 99 nodes in $G'$. 
The latency between any pair of nodes $u, u'$ among the 100 nodes is 1 time unit. 
The block validation delay is set to 0.01 time units. 
The mining rate $h_u$ of each node $u$ is 1/4000. 
At this mining rate, blocks are mined once every 40 time units on average which is roughly 10 times the broadcast delay in $G'$.
Figure~\ref{fig:non-coop-511} shows the rewards earned by the $v$, as $v$'s degree varies from 1 to 80. 
For each degree $d$, node $v$ connects to $d$ random neighbors in $G'$ (Figure~\ref{fig:noncoop_graph}).  
Each network topology is run till 10,000 blocks are mined in the blockchain. 
The experiment is repeated 5 times per topology, and over 5 different random instances of the topology. 
We see that initially node $v$'s reward increases as the degree increases. 
However, once degree reaches 25 further increase in the degree leads to a decrease in reward performance. 
\begin{figure} 
\begin{subfigure}[t]{0.47\textwidth}
\centering
\includegraphics[width=\textwidth]{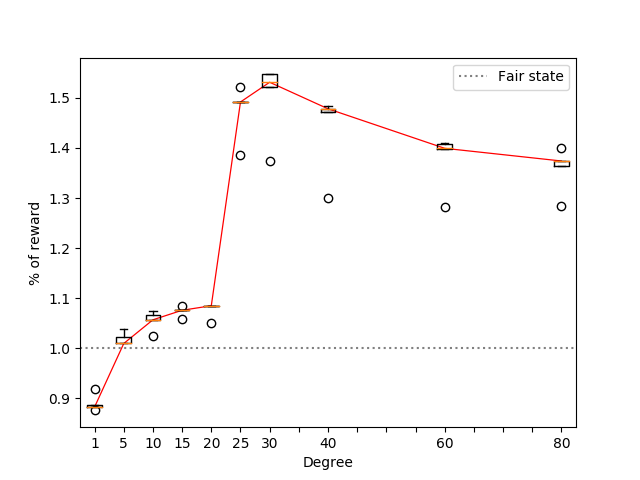} 
\caption{}
\label{fig:non-coop-511}
\end{subfigure}
\hfill 
\begin{subfigure}[t]{0.53\textwidth}
\centering
\includegraphics[width=\textwidth]{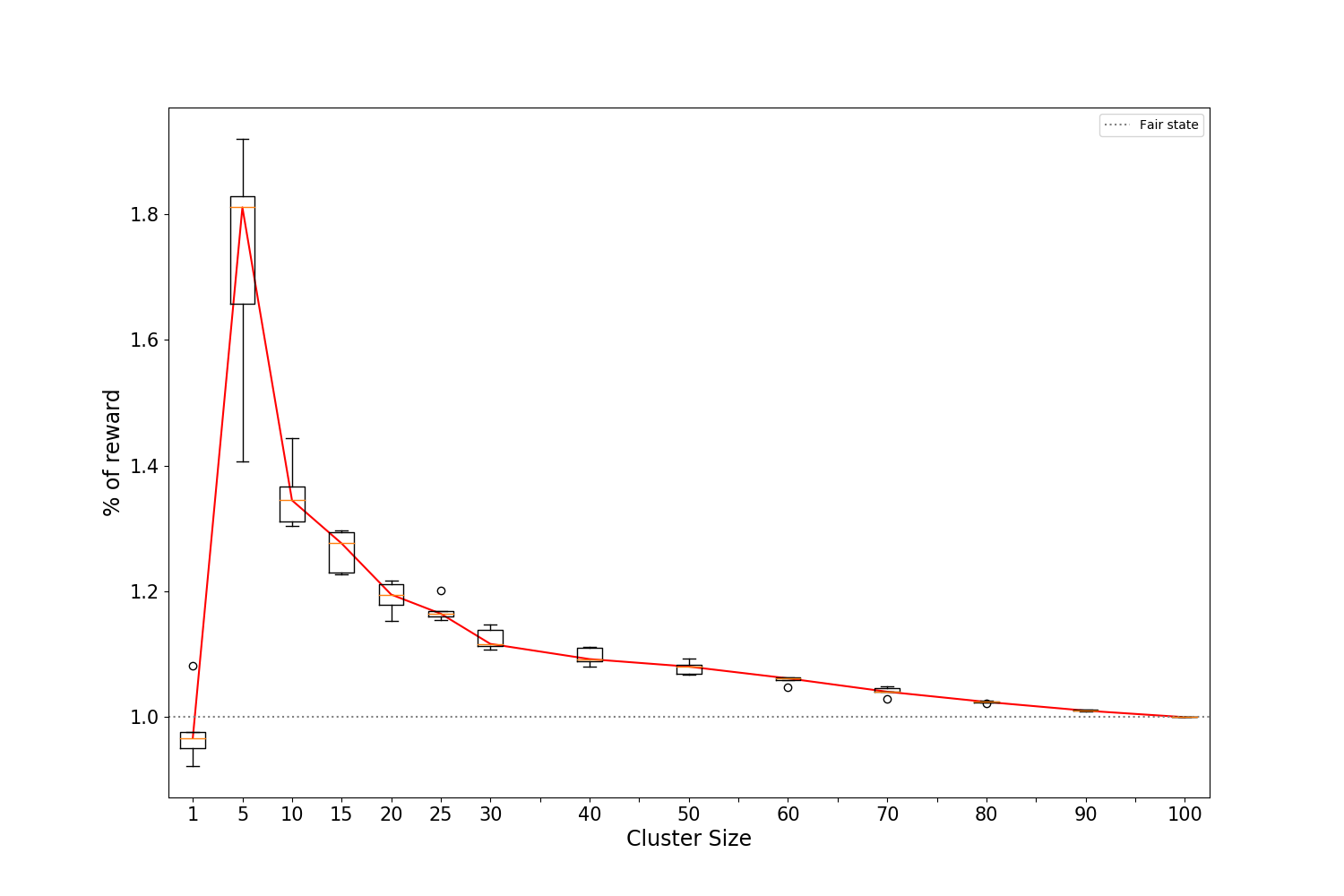} 
\caption{}
\label{fig:coop-boxplot}
\end{subfigure}
\caption{Fraction of blocks mined by $v$ in (a) non-cooperative model with increasing number of neighbors, and (b) cooperative model with increasing size of $v$'s cluster. Circle marks represent outliers.}
\label{fig:toyexam} 
\end{figure}

In the cooperative miner model, node $v$ may collude with a set of miners to maximize its reward. 
We consider the same 100 node network connected as a random graph with average degree 4 as before. 
Next, we choose a cluster of $d$ nodes including $v$ with each node connected to all other nodes in the cluster.  
The topology of the remaining $100-d$ nodes is set to a random graph of average degree 4. 
The $d$-node cluster is connected to the remaining nodes through only 2 edges as shown in Figure~\ref{fig:coop_graph}. 
Under this setting, we vary the size of the cluster $d$ and plot the average reward earned by miners within the cluster in Figure~\ref{fig:coop-boxplot}. 
Here too, we observe that as the cluster size increases initially there is an increase in the reward, but after a threshold size the reward starts decreasing. 

These observations illustrate a tight tension between increasing connectivity to reduce latency at a node and thereby indirectly helping to increase the connectivity at other nodes.  
The effect of topology on mining rewards is diminished if the mining rate is set such that the average time between mining consecutive blocks is much larger than the time taken to broadcast a block to a majority of nodes in the network.\footnote{Under an approximation, it can be shown that for a network with a given topology and set of underlying latencies, the relative change in a miner's reward due to the network topology is inversely proportional to the average time between mining consecutive blocks.} 
E.g., in Bitcoin a block is mined once every 10 minutes on average which is more than $60\times$ the block broadcast delay~\cite{decker2013information,miller2015discovering,kiffer2018better}. 
Whereas in systems like Ethereum where the average interblock time is only $10$ to $20\times$ the broadcast latency~\cite{wang2021ethna}, the topology is likely to play a significant role in the rewards obtained.\footnote{For simplicity we have not considered uncle blocks or the Kademlia overlay used in Ethereum in this paper~\cite{ritz2018impact,marcus2018low}. Factoring these aspects into the blockchain model is an avenue for future research. } 
Nevertheless, even in Bitcoin small biases in mining rewards due to network topology can add up over time resulting in significant differences in the number of blocks mined by (near-)optimally connected nodes and the others.
In the remainder of the paper, we have deliberately considered settings where the time between mining consecutive blocks is close to (less than $3\times$) the latency of broadcasting a block over the network for the purpose of clearly highlighting effects of topology on rewards. 
Our results hold even at lower mining rates, albeit with relatively smaller differences between the rewards of different miners. 
More importantly, we believe a systematic study of these effects (from both performance and security viewpoints) can help inspire a new family of highly-efficient network-aware consensus protocols that are much more scalable than today's blockchain solutions which often treat the network as a blackbox having simple high-level properties (e.g., synchrony). 

\begin{figure} 
\begin{subfigure}[t]{0.5\textwidth}
\centering
\includegraphics[width=\textwidth]{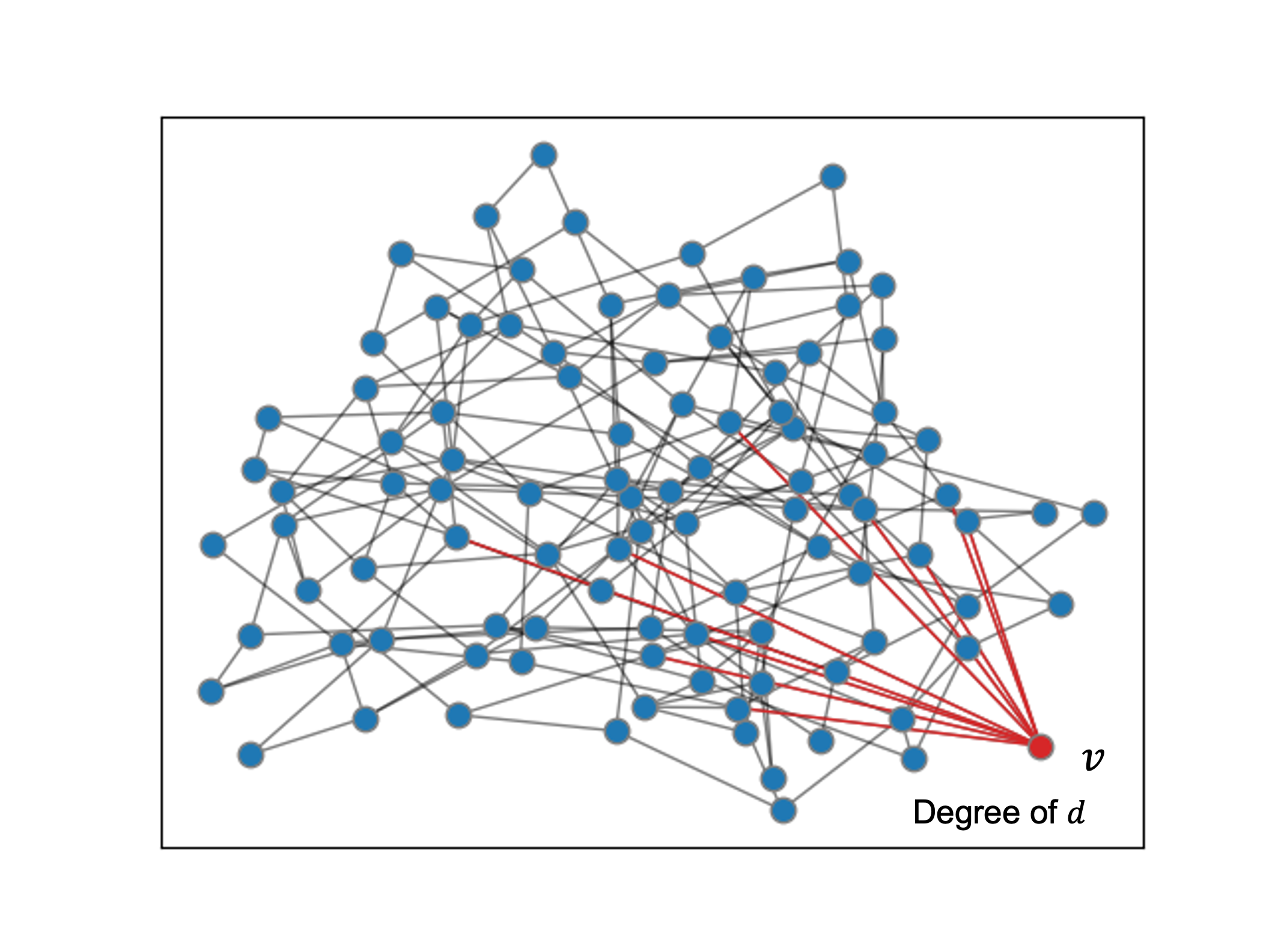} 
\caption{}
\label{fig:noncoop_graph}
\end{subfigure}
\hfill 
\begin{subfigure}[t]{0.5\textwidth}
\centering
\includegraphics[width=\textwidth]{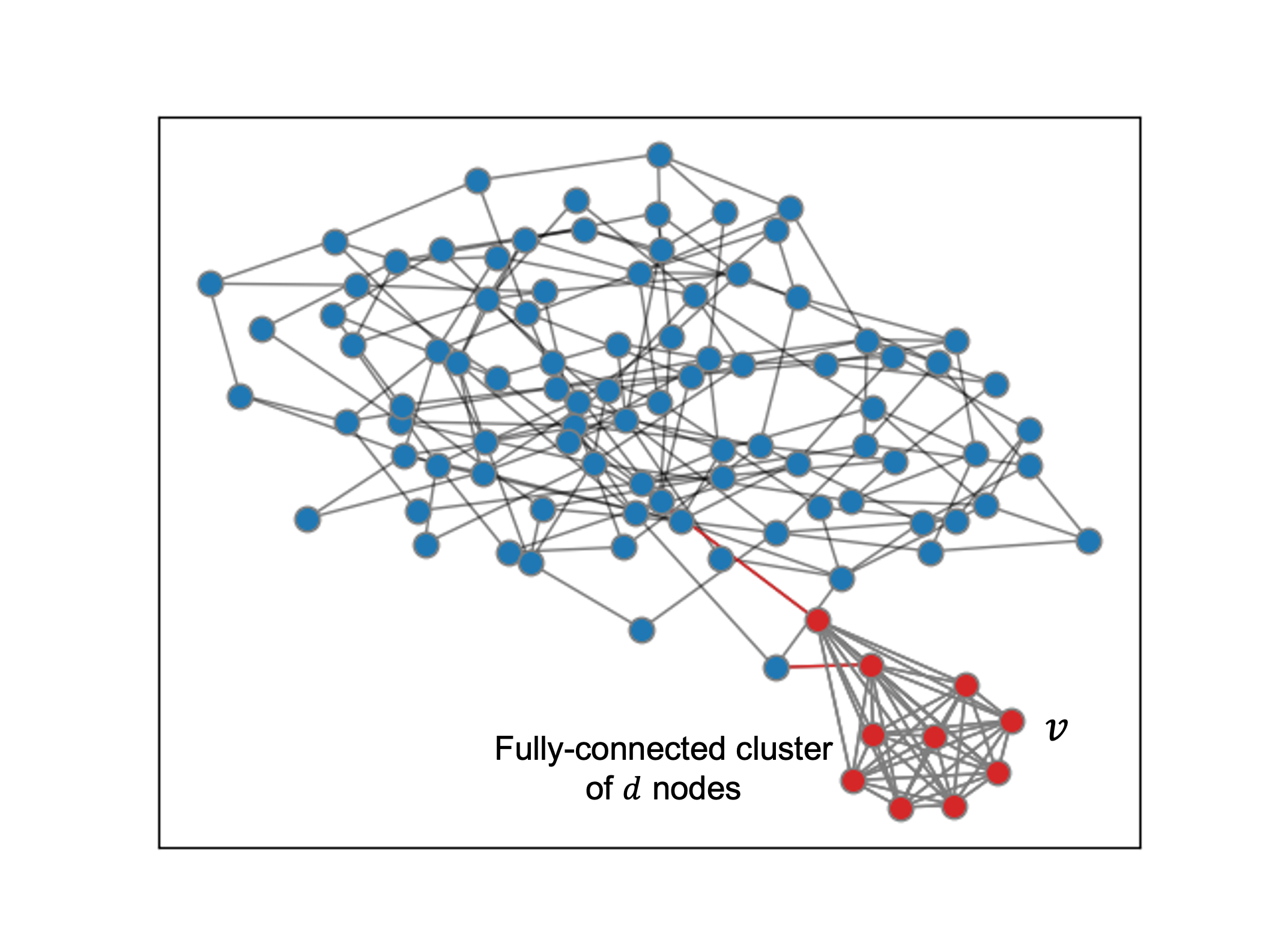} 
\caption{}
\label{fig:coop_graph}
\end{subfigure}
\caption{Example of network topology used in (a) noncooperative and (b) cooperative setting experiments of \S\ref{s:motivation}.}
\label{fig:toyexam_graph} 
\end{figure}

\section{Analysis}
\label{s: analysis}

To reason about the observations made in \S\ref{s:motivation}, we consider an analytical model of the blockchain network by which we can theoretically understand mining rewards earned under different topologies. 
The model we use in this section is simplified (vis-\`a-vis the model in Section~\ref{s: sysmodel}) to facilitate analysis.  

We consider a blockchain network as an undirected graph $G(V, E)$ with hash power $h_v = 1/|V|$ per node and link latency $l_{(u,v)} \geq 0$ for each link $(u,v) \in E$. 
For simplicity, we assume the block validation time $c_v = 0$ for all $v \in V$. 
A key difference between the model in \S\ref{s: sysmodel} and the present model, is the process of mining blocks. 
We assume blocks are mined in rounds, where in each round exactly one miner mines a block and broadcasts it over the network. 
We count rounds starting from 1 and use $\mathbb{N} = \{1, 2, 3 \ldots\}$ to denote the set of all rounds.
The miner that mines a block in a round is chosen randomly and independently of prior rounds, with $h_v$ being the probability of choosing miner $v$. 
When a miner mines a block, it broadcasts the block to the rest of the network as described in Section~\ref{s: sysmodel}.
Rounds are modeled to occur periodically, every 1 time unit. 
We refer to the genesis block as a block mined at round 0. 
A block mined during round $r$ will be referred to as block $r$. 

 For each blockchain network $G(V, E)$ we consider an associated latency graph $L(V, \tilde{E})$ comprising of miners $V$ and edges $\tilde{E}$ between every pair of miners.  
 Each edge $(u, v) \in \tilde{E}$ has a weight $\delta_{(u,v)} \geq 0$ which is the minimum time it takes for $v$ to receive a block mined by $u$ over network $G$ (or vice-versa, i.e., $\delta_{(u,v)} = \delta_{(v,u)}$). 
 The latency graph $L$ dictates the number of blocks mined by different miners that are included in the longest chain, and therefore is going to be the focus of our analysis.



For a miner $v \in V$, we define the average fraction of blocks mined by $v$ that are included in the longest chain,
\begin{align}
F_v = \lim_{R \rightarrow \infty } \frac{\sum_{r=1}^R \mathbf{1}_{\text{block } r \text{ is mined by } v \text{ and is included in the longest chain}}}{\sum_{r=1}^R \mathbf{1}_{\text{block } r \text{ is included in the longest chain}}}, \label{eq: Fv define}
\end{align}
and 
the fraction of blocks mined by $v$ that are excluded from the longest chain, i.e., 
\begin{align}
W_v = \lim_{R \rightarrow \infty} \frac{\sum_{r=1}^R \mathbf{1}_{\text{block } r \text{ is mined by } v \text{ and is excluded from the longest chain}}}{\sum_{r=1}^R \mathbf{1}_{\text{block } r \text{ is mined by } v}}. \label{eq:Wdefn}
\end{align}
as in \S\ref{s: sysmodel}, assuming the limits exist. 

\subsection{One Cluster}


\begin{figure} 
\begin{subfigure}[t]{0.45\textwidth}
\centering
\includegraphics[width=0.85\textwidth]{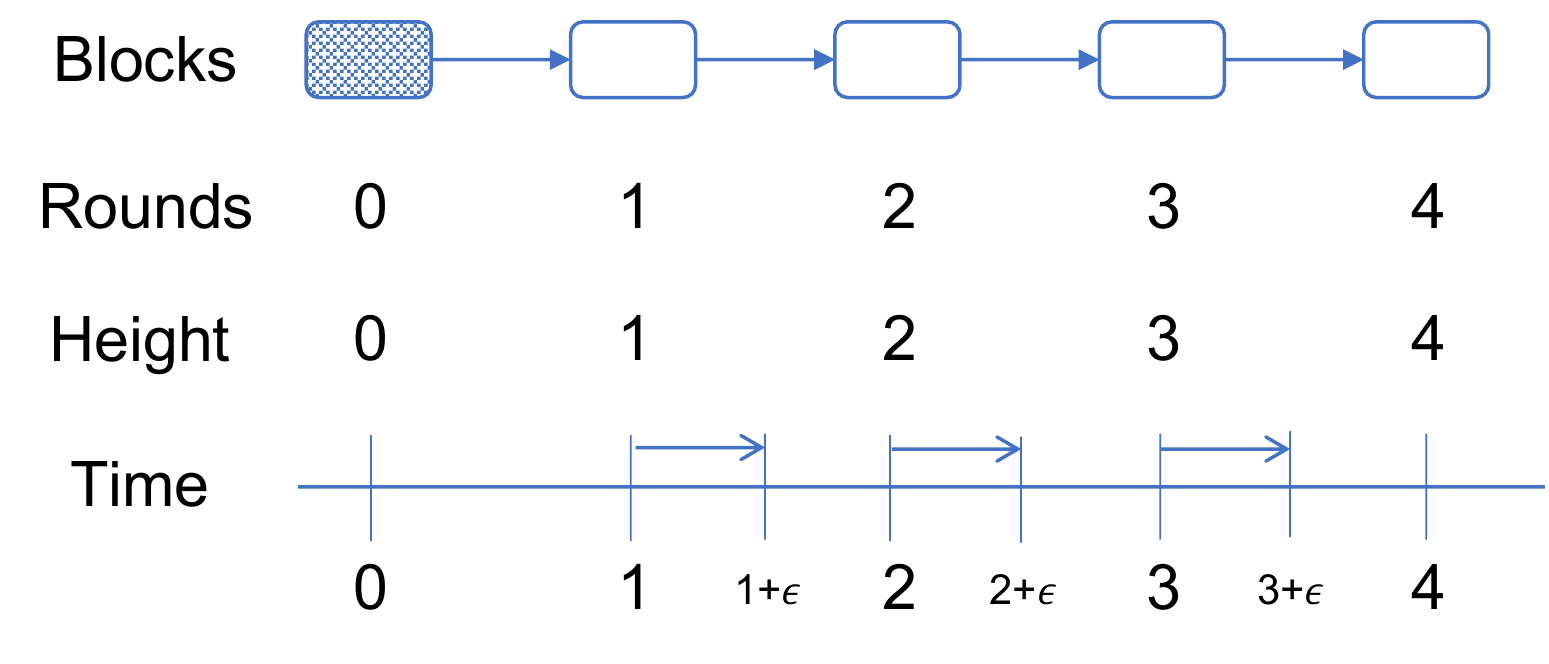} 
\caption{}
\label{fig:onecluster1}
\end{subfigure}
\hfill 
\begin{subfigure}[t]{0.55\textwidth}
\centering
\includegraphics[width=\textwidth]{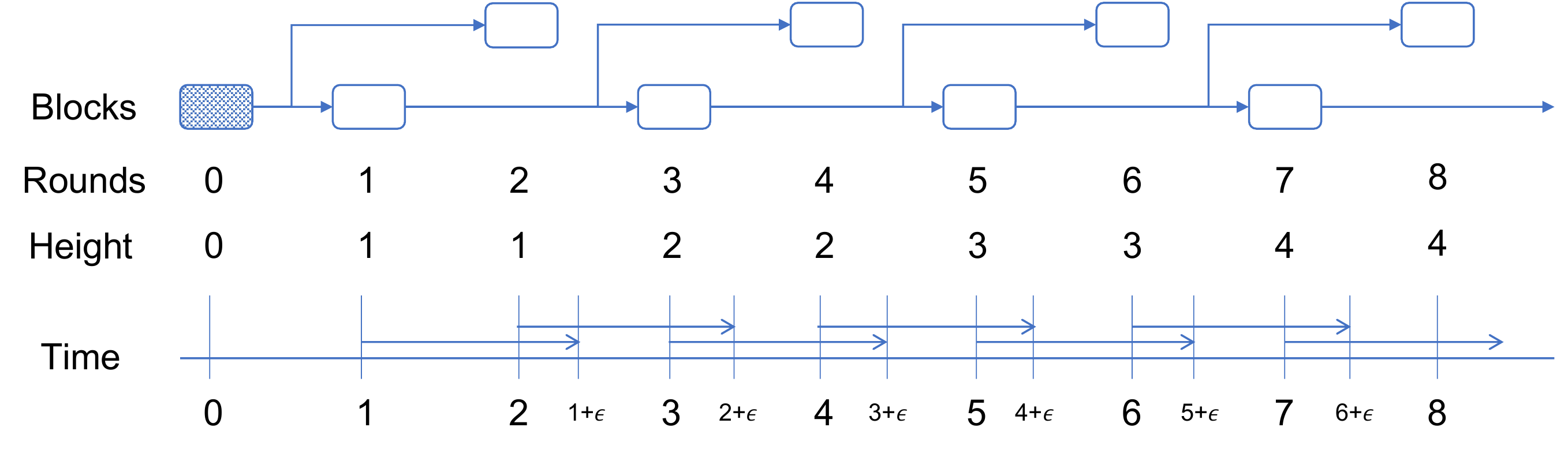} 
\caption{}
\label{fig:onecluster2}
\end{subfigure}
\caption{Evolution of the blockchain if miners are organized as a single cluster. (a) There is no forking if the link delay is less than the round length, (b) half of the blocks are in a forked chain if the link delay exceeds the round length. Time shows when each block mined reaches other nodes.}
\label{fig:onecluster} 
\end{figure}

As an illustration of the model, we first consider the simple case where the latency graph $L$ is such that $\delta_{(u,v)} = \epsilon$ for all $(u,v) \in \tilde{E}$ where $\epsilon > 0$. 
In other words, each miner is equidistant from all other miners. 
We consider two cases depending on the value of $\epsilon$: (1) $0 < \epsilon \leq 1$ and (2) $1 < \epsilon < 2$. 
We do not consider the case where $\epsilon > 2$ as it is unlikely for a blockchain system to be configured such that the average interval between successive blocks is significantly less than the time it takes for a block to propagate to a majority fraction (e.g., 90 percent) of the network.  

\smallskip 
\noindent 
{\bf Case 1:} If $\epsilon \leq 1$ then when a block is mined in a round, it is received by all other miners in the network before the start of the next round. 
Therefore any block that is mined at the next round, is going to be mined on top of the block mined in the previous round. 
This implies there would be no forks in the blockchain at any time, and a unique longest chain would grow at each miner as rounds progress. 
Figure~\ref{fig:onecluster1} illustrates this process. 

\smallskip
\noindent 
{\bf Case 2:} Next, consider the case when $1 < \epsilon < 2$. 
In this case, a block mined during a round $r \in \mathbb{N}$ will not be received by any other miner in time when the next round $r + 1$ begins. 
Whereas all blocks mined up until round $r - 1$ will be received by all the miners before the start of round $r + 1$. 
If $r+1$ is an even number, then a block mined at round $r+1$ will be mined on top of block $r-2$. 
Whereas if $r+1$ is an odd number, then block $r+1$ will be mined on top of block $r-1$. 
This is because for an even $r+1$ blocks $r-2$ and $r-1$ are both at the same height in the blockchain. 
However, block $r-2$ reaches the miner of round $r+1$ before block $r-1$. 
Therefore block $r+1$ will be mined on top of block $r-2$. 
For an odd $r+1$, block $r-1$ will be at a height that is strictly greater than the height of block $r-2$. 
Hence block $r+1$ will be mined over block $r-1$ in this case. 
Figure~\ref{fig:onecluster2} illustrates how the blockchain evolves when $1 < \epsilon < 2$. 
Thus, at each height of the blockchain there is a fork and only half of all blocks mined are included in the longest chain. 

The two cases outlined above motivate the following theorem. 
We omit the proof as it is straightforward. 
\begin{thm} \label{thm:singcluster}
For a single cluster latency graph $L$ with inter-miner latency $\epsilon$, we have \\
(i) if $0 < \epsilon \leq 1$, then for any $v$ we have $\mathbb{E}[F_v] = 1/n$ and $\mathbb{E}[W_v] = 0$; \\
(ii) if $1 < \epsilon < 2$, then for any $v$ we have $\mathbb{E}[F_v] = 1/n$ and $\mathbb{E}[W_v] = 1/2$. 
\end{thm}

The result in Theorem~\ref{thm:singcluster} also holds if the shortest path latencies are not all the same, but within the same range of values: if $0 < \delta_{(u,v)} < 1$ for all $u, v$ then claim (i) in the Theorem holds, and if $1 < \delta_{(u,v)} < 2$ for all $u, v$ then claim (ii) holds. 
E.g., a random topology is likely to produce shortest path latencies that are roughly similar for all miners, if the miners are all concentrated within a single geographic region (e.g., same continent or country).  
If the cost of mining blocks is negligible compared to rewards earned, Theorem~\ref{thm:singcluster} says that it does not matter what degree miners choose (so long as the network is connected). 
Whereas if the reward depends on the number of blocks mined, then a higher degree is desirable as it would result in a smaller latency value and minimize forking. 

\subsection{Two Clusters}
\label{sec: two clusters}

Next, we consider the case where there are two clusters of miners, where miners in each cluster have low-latency paths to each other but a high latency to miners of the other cluster. 
This can happen, e.g., if within a cluster miners are densely connected, whereas links between the clusters are sparse. 
To model this, we suppose the latency graph $L$ comprises of two sets $V_1$ and $V_2$ of miners with $|V_1| = pn$ and $|V_2| = (1-p)n$ where $0.5 < p < 1$. 
The latency between any two miners $u, v \in V_1$ is $\delta_{(u,v)} = \epsilon$ where $0 < \epsilon < 1$. 
Similarly, the latency between any two miners $u, v \in V_2$ is also $\delta_{(u,v)} = \epsilon$ where $0 < \epsilon < 1$. 
However, the latency between a miner $u \in V_1$ and a miner $v \in V_2$ is $\delta_{(u, v)} = \Delta$, where $1 < \Delta < 2$.
We assume $\Delta - 1 > \epsilon$, to model the case where the inter-cluster miner latencies are significantly larger than the intra-cluster miner latencies. 
This assumption implies if block $i$ is mined by cluster 1 and block $i+1$ is mined by cluster 2, then the $(i+1)$-th block is propagated to all miners in cluster 2 before they receive the $i$-th block (and vice-versa if the $i$-th block is mined by cluster 2 and the $(i+1)$-th block is mined by cluster 1). 

\begin{figure}[t]
    \centering
    \includegraphics[width=0.98\textwidth]{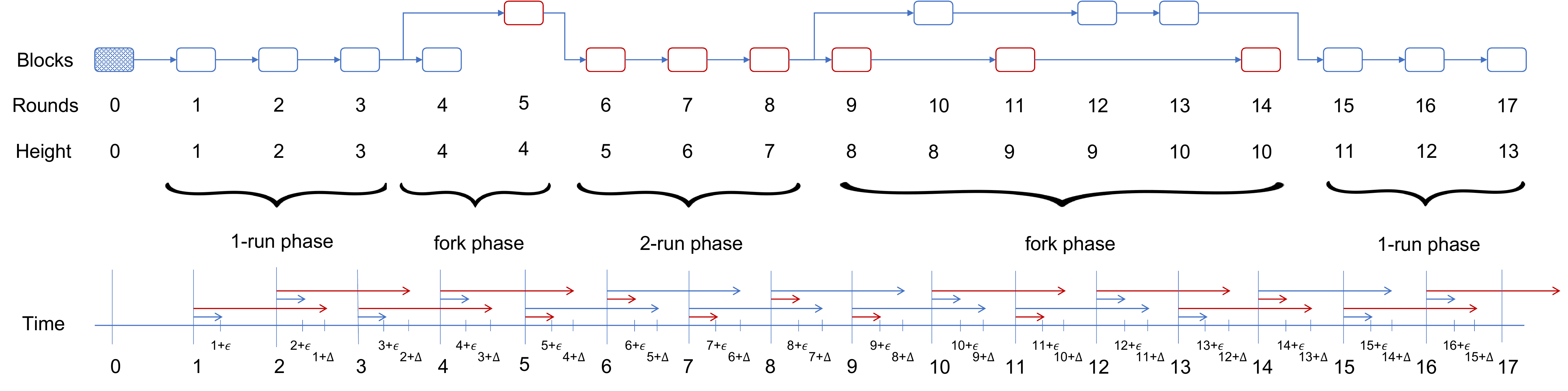}
    \caption{Evolution of the blockchain when there are two clusters. Blocks in blue are mined by cluster 1 miners, while blocks in red are mined by cluster 2 miners. The three types of phases  (1-run, 2-run and fork) that occur during the evolution are also marked. Time shows when a block mined reaches miners in cluster 1 (blue arrow) and miners in cluster 2 (red arrow).}
    \label{fig:twocluster_blockdiagram}
\end{figure}

We analyze this case by first observing that there are distinct {\em phases} that occur in any sample path of blockchain evolution. 
To illustrate this, consider the example shown in Figure~\ref{fig:twocluster_blockdiagram} where the sequence of clusters that mine blocks in each round are as follows: $1, 1, 1, 1, 2, 2, 2, 2, 2, 1, 2, 1, 1, 2, 1, 1, 1$, i.e., blocks 1--4 are mined by cluster 1, block 5--7 are mined by cluster 2 etc.
For this example, the blockchain evolves as shown in Figure~\ref{fig:twocluster_blockdiagram}. 
The different phases in this evolution are as follows. 

For the first three rounds, there are no forks. 
This is because the block mined in round 1 by a miner in cluster 1 reaches the miner of second block within $\epsilon$ time units, allowing the second miner to build its block on top of the first block. 
Similarly the third block is built on top of the second block.
Each of these three blocks is part of the longest chain. 
We call such a sequence of contiguous blocks that are all mined by miners within cluster 1 as a {\em 1-run} phase.
Similarly, if there is a continuous sequence of blocks that are all mined by miners of cluster 2 (e.g., blocks 6, 7, 8), we call those blocks as belonging to a {\em 2-run} phase. 
It can be seen that blocks mined during a 2-run phase also belong to the longest chain. 

Next, consider a sequence such as $2, 1, 2, 1, 1, 2$ where every two consecutive blocks is mined by an opposing pair of clusters such as $2, 1$ or $1, 2$. 
We call such a phase as a {\em fork phase}. 
From Figure~\ref{fig:twocluster_blockdiagram}, we see that during a fork phase the blockchain has two competing forks that grow in parallel. 
We formalize this observation in the following lemma. 
\begin{lem} \label{lem:twoclus}
For a two cluster network with a latency graph as discussed in this section (\S\ref{sec: two clusters}), during a fork phase there are two competing forks with miners in each cluster growing a distinct fork. 
\end{lem} 
(Proof in Appendix~\ref{apx:lemmatwoclusproof}). 

A fork phase ends when there are two consecutive blocks mined by the same cluster. 
E.g., if $c_1, c_2, \ldots, c_{2i}$ is a sequence of clusters that mined blocks during a fork phase and if $c_{2i+1}, c_{2i+2}$ are the clusters that mined the two blocks following the fork phase, then if $c_{2i+1} = c_{2i+2}$ the fork phase ends. 
When a fork phase ends, among the two competing forks during the fork phase one of them ``wins out'' and become part of the longest chain. 
If $c_{2i+1} = c_{2i+2} = 1$ then the fork of cluster 1 wins out, whereas if $c_{2i+1} = c_{2i+2} = 2$ the the cluster 2's fork wins out (see Figure~\ref{fig:twocluster_blockdiagram}).
Thus when a fork wins out at the end of a fork phase, all blocks mined by the winning cluster become part of the longest chain.  
Note that is is possible for a new fork phase to begin starting with the second round after a fork phase ends (i.e., from the round of $c_{2i+2}$).  

To compute $\mathbb{E}[F_v]$ and $\mathbb{E}[W_v]$ for a miner $v$ in this two cluster setting, for any sequence of blocks mined $b_1, b_2, \ldots$ in each round starting from the genesis block, let $P_1, P_2, \ldots$ denote the sequence of phases that occur across the rounds. 
E.g., in Figure~\ref{fig:twocluster_blockdiagram}, $P_1$ is a 1-run phase that lasts during rounds $1, 2, 3$; 
$P_2$ is a fork phase that lasts during rounds $4, 5$; 
$P_3$ is a 2-run phase that lasts during rounds $6, 7, 8$; 
$P_4$ is a fork phase in rounds $9, 10, 11, 12, 13, 14$; 
$P_5$ is a 1-run phase during round $15, 16, 17$. 
For a miner $v \in V$, let $P^v_i$ be the number of blocks mined by $v$ during phase $P_i$ that is included in the longest chain. 
Similarly, let $P^c_i$ be the number of blocks mined by cluster $c \in \{1, 2\}$ during phase $P_i$ that is included in the longest chain. 
Supposing the network runs for $k$ phases. 
Then, the expected number of blocks $M_c$ mined by cluster $c$ that are in the longest chain over the $k$ phases is 
\begin{align}
M_c = \sum_{i=1}^k P^c_i \Rightarrow \mathbb{E}[M_c] = \sum_{i=1}^k \mathbb{E}[P^c_i]. 
\label{eq:phasesum}
\end{align}
Since the phases are identically distributed, it suffices to compute $\mathbb{E}[P^c_1]$ to evaluate $\mathbb{E}[M_c]$ in Equation~\eqref{eq:phasesum}.
Using $\mathbb{E}[M_c]$ we can compute $\mathbb{E}[F_v]$ as shown below. 
\begin{thm} 
\label{thm:twocluster}
For a two cluster latency graph $L$ discussed in this section (\S\ref{sec: two clusters}) with $|V_1| = p |V|$ and $|V_2| = (1-p)|V|$, 
we have $\mathbb{E}[F_v] = $ 
\begin{align}
\frac{1}{n} \left[ \frac{p}{1-p} + \frac{2p^2(1-p)}{(1-2p(1-p))^2} \right] / \left( \left[ \frac{p^2}{1-p} + \frac{2p^3(1-p)}{(1-2p(1-p))^2} \right] \right. \notag \\ 
+ \left. \left[ \frac{(1-p)^2}{p} + \frac{2(1-p)^3p}{(1-2p(1-p))^2} \right] \right)
\end{align}
for any $v \in V_1$. 
\end{thm}
(Proof in Appendix~\ref{apx:thmtwoclusproof}). 

\begin{figure} \begin{subfigure}[t]{0.5\textwidth}
\centering
\includegraphics[width=0.75\textwidth]{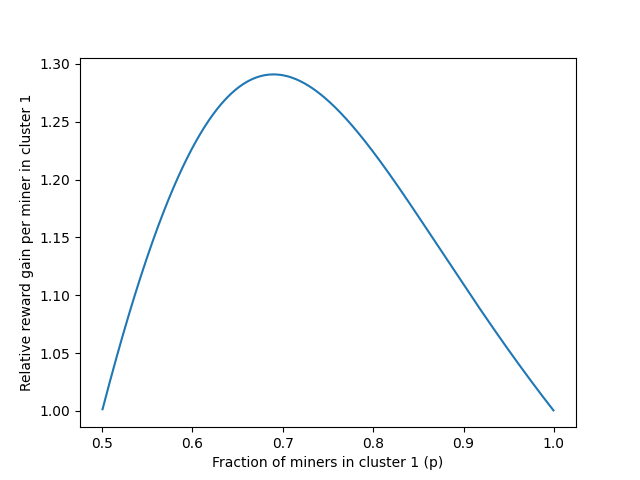} 
\caption{}
\label{fig:twoclustergain}
\end{subfigure}
\hfill 
\begin{subfigure}[t]{0.5\textwidth}
\centering
\includegraphics[width=0.75\textwidth]{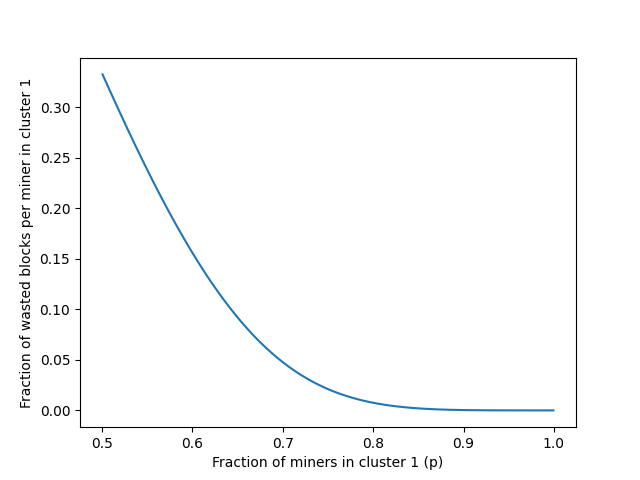} 
\caption{}
\label{fig:twoclusterwastage}
\end{subfigure}
\caption{(a) Relative gain over the fair value (i.e., $\mathbb{E}[F_v] / (1/n)$) of a miner $v$ in the dominant cluster with varying cluster size. (b) Percentage of blocks mined by $v$ that are not in the longest chain as a function of cluster size. }\label{fig:twoclusterplots} 
\end{figure}


Thus, the percentage of blocks mined by a miner in $V_1$ could be significantly different from the fair allocation of $1/n$ in this two cluster case. 
In Figure~\ref{fig:twoclustergain} we plot the ratio between $\mathbb{E}[F_v]$ as in Theorem~\ref{thm:twocluster} and the fair allocation of $1/n$ as a function of the cluster size parameter $p$. 
We see that the ratio is maximum at $1.29$ when the cluster size is roughly $69 \%$ of the total number of miners in the network. 

Theorem~\ref{thm:twocluster} implies that it is beneficial for miners to form connections with others on the network in such a way that there is a well-connected `dominant' cluster. 
In general, the optimal size and maximum possible reward for the dominant cluster depends on the underlying link latencies and connectivity graph among peers outside of the cluster.  
E.g., in Appendix~\ref{apx:optclusdelta} we show that as the latency $\Delta$ increases, the optimum size of the dominant cluster drops. 
The per-miner gain diminishes if the size of the dominant cluster either increases or decreases from this optimal value.
This is because with a small cluster size (i.e., $<$ optimal), miners in the dominant cluster have low-latency paths to only a small number of nodes resulting in poor reward improvements. 
With a large cluster size ($>$ optimal), despite a miner in the dominant cluster having low-latency paths to many nodes, it does not have adequate competitive edge as a large number of other miners also have low-latency paths to many nodes. 
At the optimum cluster size, a dominant cluster miner achieves the best balance between having low-latency paths to an adequate number of miners while ensuring the number of other well-connected miners is limited.
For miners that have been excluded from this dominant cluster, we show in Appendix~\ref{s:threeclusters} that being well-connected with each other is beneficial even if they are not able to achieve a low latency path to the dominant cluster.  
We now present the fraction of blocks mined by a miner that are included in the longest chain.  
\begin{thm}
\label{thm:fprimetwoclus}
For a two-cluster network with latency graph as in this section (\S\ref{sec: two clusters}), we have 
\begin{align}
\mathbb{E}[W_v] = \left( \frac{2(1-p)^3}{(1-2p(1-p))^2} \right) / \left( \frac{p}{1-p} + \frac{2p^2(1-p)}{(1-2p(1-p))^2} +  \frac{2(1-p)^3}{(1-2p(1-p))^2} \right).  
\end{align}
\end{thm}
(Proof in Appendix~\ref{apx:fprimetwoclusproof}).  

Figure~\ref{fig:twoclusterwastage} plots $\mathbb{E}[W_v]$ for a miner $v \in V_1$ as a function of $p$.
Unlike $\mathbb{E}[F_v]$, we see that $\mathbb{E}[W_v]$ is monotonically decreasing with $p$. 
When $p = 0.7$ the wastage fraction is already less than 5 \%. 
Thus, depending on the relative magnitude of the rewards obtained per mined block included in the longest chain and cost incurred per mined block, miners may be incentivized to form a cluster of an appropriate size. 

\smallskip 
\noindent 
{\bf Implications for cooperative and noncooperative topology designs.}
The noncooperative model can be understood as an instance of the two cluster scenario considered in this section, where the first cluster has only one node (i.e., $p=1/n$). 
If the node $v \in V_1$ is poorly connected to $V_2$ with $2 > \Delta > 1 + \epsilon$, Theorem~\ref{thm:twocluster} implies that $\mathbb{E}[F_v] = O(1/n^2)$. 
On the other hand, if $v \in V_1$ is well connected to $V_2$ with $\Delta < 1$, then $\mathbb{E}[F_v] = 1/n$ from Theorem~\ref{thm:singcluster}. 
Importantly, increasing $v$'s connectivity beyond this point to further reduce $\Delta$ does not increase $v$'s average reward. 
In the cooperative model, we have already discussed how an optimal cluster size leads to maximum rewards in Theorem~\ref{thm:twocluster}.  
How to achieve optimal connectivity in these models using distributed algorithms, with limited knowledge of the global network topology and latencies is an important question which we defer for future work.  
Adaptively trying different cluster sizes guided by measurement of rewards gained in real time may be a possible approach~\cite{mao2020perigee}, but could be challenging due to the sparsity and high-variance in the reward feedback.

\section{Evaluation}
\label{s: evaluation}
In this section we present experimental evaluation to highlight reward bias due to network topology and geographical location 
in real-world settings.
We also present candidate protocols that miners can adopt in order to alleviate the bias and increase their profits. 

\subsection{Simulation Setup}
\label{s: eval model}
We conduct our experiments on OMNeT++ 5.6.2, an event driven simulator~\cite{varga2008overview}, by implementing a detailed mining and block propagation model based on the Bitcoin protocol. 
There are 246 miners in our simulation, with each miner mining blocks, broadcasting messages and maintaining a local replica of the blockchain in a completely decentralized fashion. 
The model we have used in the simulations builds upon the model described in \S\ref{s: sysmodel}. 

\smallskip
\noindent
\textbf{Mining.}
We apply Bitcoin's mining and broadcasting strategy in our simulation.
At the start of the simulation, all the nodes mine the block on top of the genesis block.
To simulate mining, each miner has an independent random timer that expires according to an exponential distribution with a default mean of 15 seconds. 
We use the same rate for all miners, to model uniform hash power across all miners. 
When a miner's timer expires, it generates a block on top of the block with the greatest height on its blockchain, and broadcasts it immediately to its connected neighbors. 
When a miner is waiting for its timer to expire, if it receives a new block that extends its current longest chain, it updates it local blockchain and resets the timer to mine on top of the new received block.

\smallskip
\noindent
\textbf{Block broadcast.}
Miners receive and generate messages of different types, each of which causes the state of the miner to be updated as follows. 
\begin{itemize}
\item[*] A block message has an identifier and a reference to the identifier of the block's parent. 
We do not include any transactions in a block message; nor do we model bandwidth limitations in the network. 
When a miner receives a block message, it first validates the block by verifying it points to a valid parent block in its blockchain.
If the received block is valid and extends the longest chain in the blockchain by one, the miner adds the block in to its blockchain. 
It then resets its exponential timer to mine a new block over the latest received block.
We set a duration of 1 ms at each miner for this validation operation. 
Upon validation, the miner immediately sends out the block message to all of its neighbors except the one it received it from. 
\item[*] A miner may also receive a block message, whose parent block is unavailable at the miner's local blockchain replica. 
In this case, the miner sends a `require' message back to the sender of the block requesting to also send the parent block. 
The received block is meanwhile stored within an orphan block list at the miner. 
We set a duration of 1 ms at each miner for this operation as before. 
When a miner receives a require message, it sends back the requested block as a `response' message.
In case the response block message received to the require message is also such that its parent block is unavailable, a second require message is sent. 
This process continues until all blocks from the genesis to the first received block are available. 
\item[*] If a miner receives a block that is not part of the longest chain, it adds the block in to the blockchain, and does not relay the block to any of the neighbors. 
\end{itemize}


\smallskip 
\noindent
\textbf{Reward.}
Miners earn rewards by mining blocks that get included in the longest chain.
We measure reward by computing the percentage fraction of blocks mined by a miner that is included in the longest chain, i.e., the $F_v$ values for each miner $v$ as defined in Equation~\eqref{eq: Fv define cont}. 
While it is likely that the blockchain replicas at different miners are not completely consistent with each other, it is typically the case that the longest chain except for the last few blocks is consistent across miners. 
So we discard the last 100 blocks in the longest chain, and collect the percentage of valid blocks among the remaining blocks in the longest chain for each miner.
Each experiment is run for a sufficient duration that more than 100,000 blocks are present in the longest chain at the end of the simulation. 
This ensures there are at least a 100 blocks mined per miner within the longest chain. 
For each network topology, we run 10 independent simulations to gather confidence in the reported results. 

\smallskip
\noindent
{\bf Network model.} 
We consider 246 miners distributed across cities around the world. 
There are 94 miners from Europe, 83 from North America, 37 from Asia, 12 from South America, 11 from Africa and 9 from Australia.
This geographical distribution of miners is representative of the Bitcoin network's miner distribution~\cite{bitnodes} that has a majority of nodes at Europe and North America (Appendix~\ref{apx:simdetails}). 
The propagation delay $l_{(u,v)}$ of sending a block from any miner $u$ to miner $v$ is set using latencies of sending pings from city $u$ to city $v$ as measured in the global ping statistics dataset collected on July 19th \& 20th, 2020 from WonderNetwork~\cite{wondernetwork}.
We take half of the average measured ping round-trip times as the propagation delay between the two miners. 
The median link propagation delay in the dataset is 69 ms. 
We do not model bandwidth at the miners, and assume block sizes to be small relative to the bandwidth. 
We have set the delay for validating blocks (i.e., $c_v$ for miner $v$ in \S\ref{s:netmodel}) to be a low value of 1 ms, as in reality miners (e.g., in Bitcoin) have a significant compute power available to them~\cite{gervais2015tampering}.
The primary broadcast delay bottleneck in our experiments is therefore the 
WAN propagation delay. 
We will open-source all code and datasets used in the experiments, and include a weblink for them in the final version of the paper. 

\subsection{Results}

In the following sections (\S\ref{s:rewbiarantop}, \S\ref{s:topincrew}, \S\ref{s:allrewbias}), 
we consider various candidate topologies for the network and measure the effect they have on the rewards of miners. 
The purpose of these experiments is to solely understand these effects---we do not consider how the topologies may be constructed (desirably in a distributed fashion) by the miners or the effect of a dynamically changing topology and leave these issues for future work.  

\subsubsection{Reward-Bias on a Random Topology} 
\label{s:rewbiarantop}

\begin{figure}[t]
    \centering
    \includegraphics[width=\textwidth]{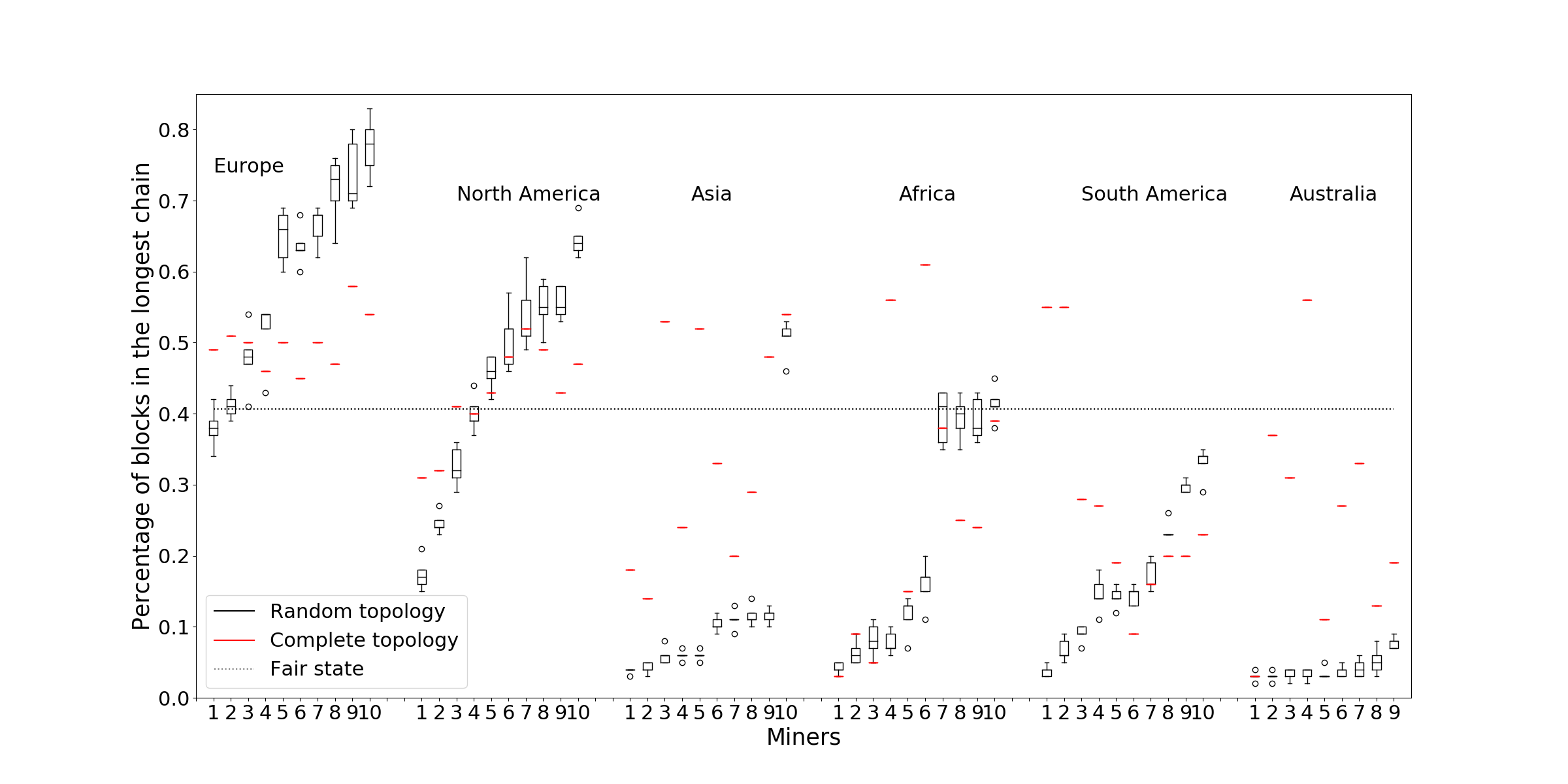}
    \caption{Fraction of blocks in longest chain as mined by miners in different continents under random and complete graph topologies. The confidence intervals are omitted for the complete graph results for clarity. }
    \label{fig:randomandcom}
\end{figure}

On a world-wide network, the distribution of where miners are located implicitly creates a clustering in the network, even if the topology is chosen in a geography-independent way (e.g., random). 
So before we provide any optimized topology designs, we first investigate how much benefit/hurt miners receive owing to their geography. 
For this experiment, we consider the 246 miner network outlined in \S\ref{s: eval model}, with link latencies following the WonderNetwork ping dataset.  
We consider a random topology and the complete graph topology.
Similar to the topology followed by Bitcoin network today, the random topology randomly picks a fixed number of neighbors for each node.
Due to the smaller size of our network, we set a (out-)degree of 6 for each miner rather than 8 as followed by Bitcoin~\cite{miller2015discovering}.
Each miner reaches out to 6 random miners and builds connections to them; it also accepts any incoming connection requests from other miners.
The complete graph topology connects each miner to the remaining 245 miners. 
Thus each miner has a direct connection to every other miner.


Figure~\ref{fig:randomandcom} shows the results, where we have plotted the percentage of blocks in the longest chain mined by miners in each of the 6 continents.
For brevity, we have randomly chosen 10 miners from each continent and have included only their results in the figure. 
We see that the Europe and North America are the dominant clusters.
In the 246 node network, there is a fair value of around 0.4\% ($=100/246$) with uniform hash power.
However, most of the European and North American miners get a greater percentage than the fair value while miners in other continents get below their fair share in both topologies.
In the complete graph topology, Europe and North America get slightly closer to the fair state, but there is still a large difference between them and other regions.

We can explain the difference between Europe, North America and the other regions using the theoretical results of Section~\ref{s: analysis}.
As a large cluster with low link delay within the cluster, Europe and North America can get a higher acceptance rate in their mined blocks, as it corresponds to there being a single dominant cluster in the network. 
Since Europe has 94 miners and North America has 83 miners, together they have 72\% of all miners in the network, which provides strong percentage gains for every miner in the cluster.

\subsubsection{Topologies to Increase Rewards}
\label{s:topincrew}

\begin{figure} \begin{subfigure}[t]{0.4\textwidth}
\centering
\includegraphics[width=\textwidth]{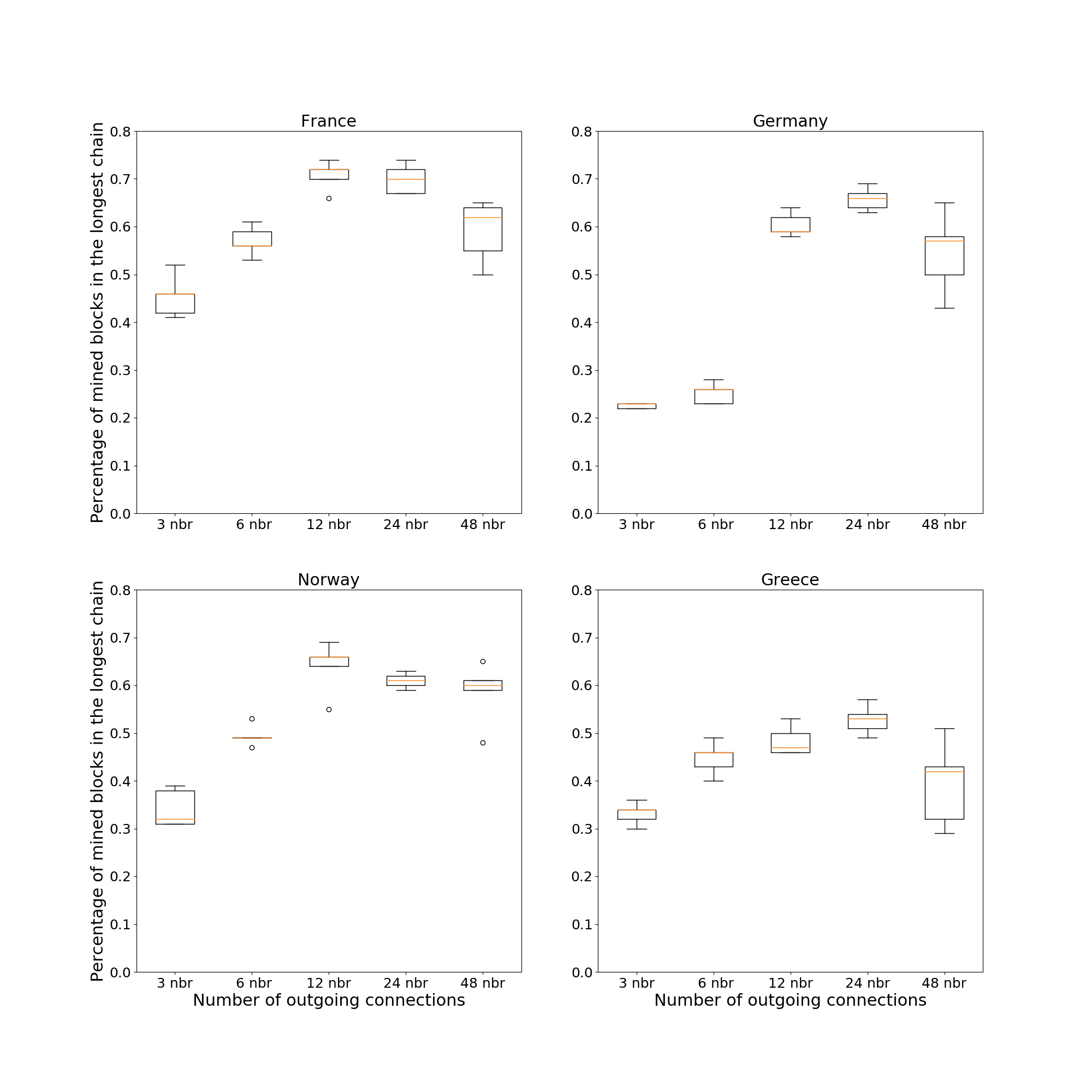} 
\caption{}
\label{fig:EuropeIncrease}
\end{subfigure}
\hfill 
\begin{subfigure}[t]{0.6\textwidth}
\centering
\includegraphics[width=\textwidth]{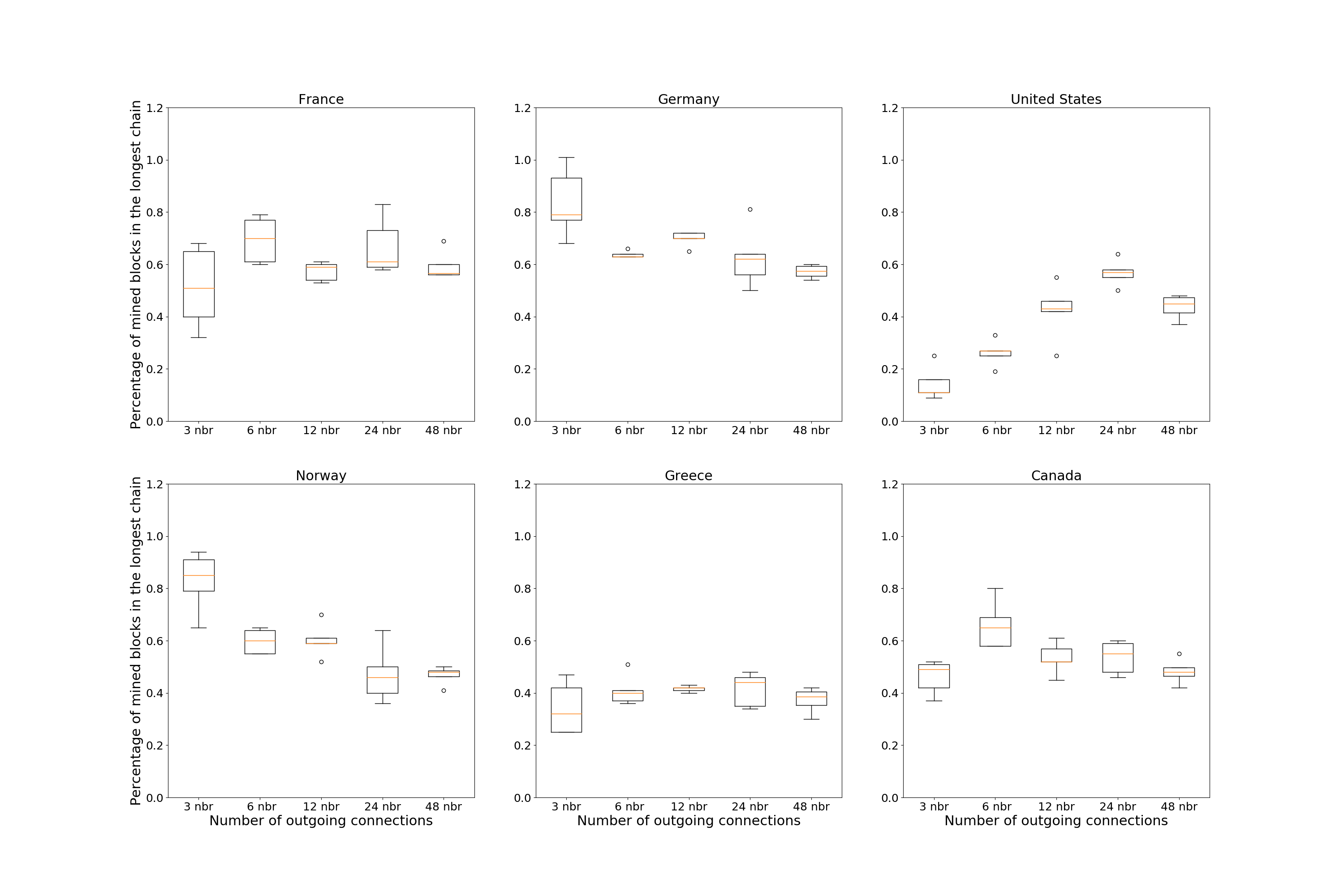} 
\caption{}
\label{fig:centrlComplete}
\end{subfigure}
\caption{(a) Increasing the degree of miners in Europe hurt them at high degree (b) Increasing the degree of miners in Europe and North America. }
\label{fig:EuNA} 
\end{figure}

Geography-induced clustering under a random topology provides a natural advantage to the European and North American nodes.
We verify whether increasing the number of neighbors further benefits miners in Europe and North America. 
By keeping the degree at other miners fixed at 6, we increase the number of random connections at European miners to vary as 3, 6, 12, 24 and 48. 
Figure~\ref{fig:EuropeIncrease} plots the performance for 4 randomly chosen miners in Europe. 
While the performance increases initially with increasing degree, after about a 12 neighbors any further increase in number of neighbors hurts the miners. 

Since there are 2 central regions, we also consider a policy where we increase the neighbor degree at both European and North American miners, 
while keeping the degree of other miners fixed. 
Figure~\ref{fig:centrlComplete} plots performance of 4 randomly picked nodes from Europe as in Figure~\ref{fig:EuropeIncrease} and 2 randomly picked nodes from North America.
Here too, the performance tends to 
 drop with more neighbors.

Next, we consider topologies in which a subset of miners are tightly interconnected to form a dominant cluster. 
We first implement a clustering policy in which nodes are randomly split into two groups, containing 70\% and 30\% of miners respectively. 
The first group is the dominant cluster, while the second group is the non-dominant cluster.
We connect each miner to every other miner within its own cluster, 
whereas we leave only 20 connections overall between the 2 clusters.
Figure~\ref{fig:70cluster} shows the results, where we only plot a random subset of miners from each group for brevity as before.
We observe that miners from the large cluster get a greater percentage of blocks because of being part of the larger cluster.
Whereas, even European and North American miners have a hard time achieving the fair state if they are in the smaller cluster. 

We also consider another clustering policy where we take Europe and North America together as the large cluster and the remaining continents as the small cluster.
As before, within a cluster we connect each miner to all other miners in the cluster, and leave  only 20 connections between clusters.
In Figure~\ref{fig:EUNACluster}, the percentage distributions are consistently above the fair state for miners in the dominant cluster compared to Figure~\ref{fig:70cluster}. 


\subsubsection{Alleviating Reward Bias}
\label{s:allrewbias}

\begin{figure} \begin{subfigure}[t]{0.5\textwidth}
\centering
\includegraphics[width=\textwidth]{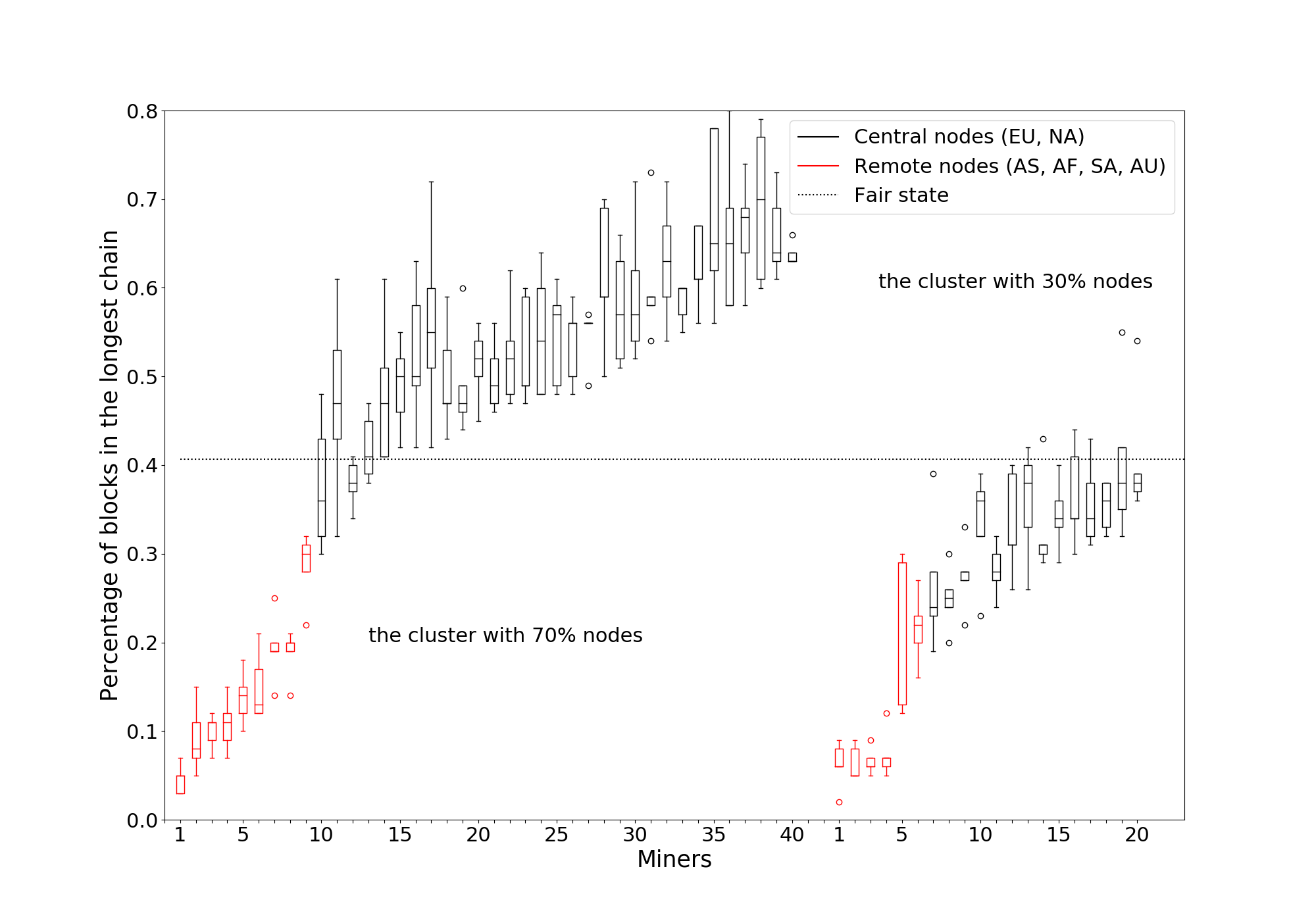} 
\caption{}
\label{fig:70cluster}
\end{subfigure}
\hfill 
\begin{subfigure}[t]{0.5\textwidth}
\centering
\includegraphics[width=\textwidth]{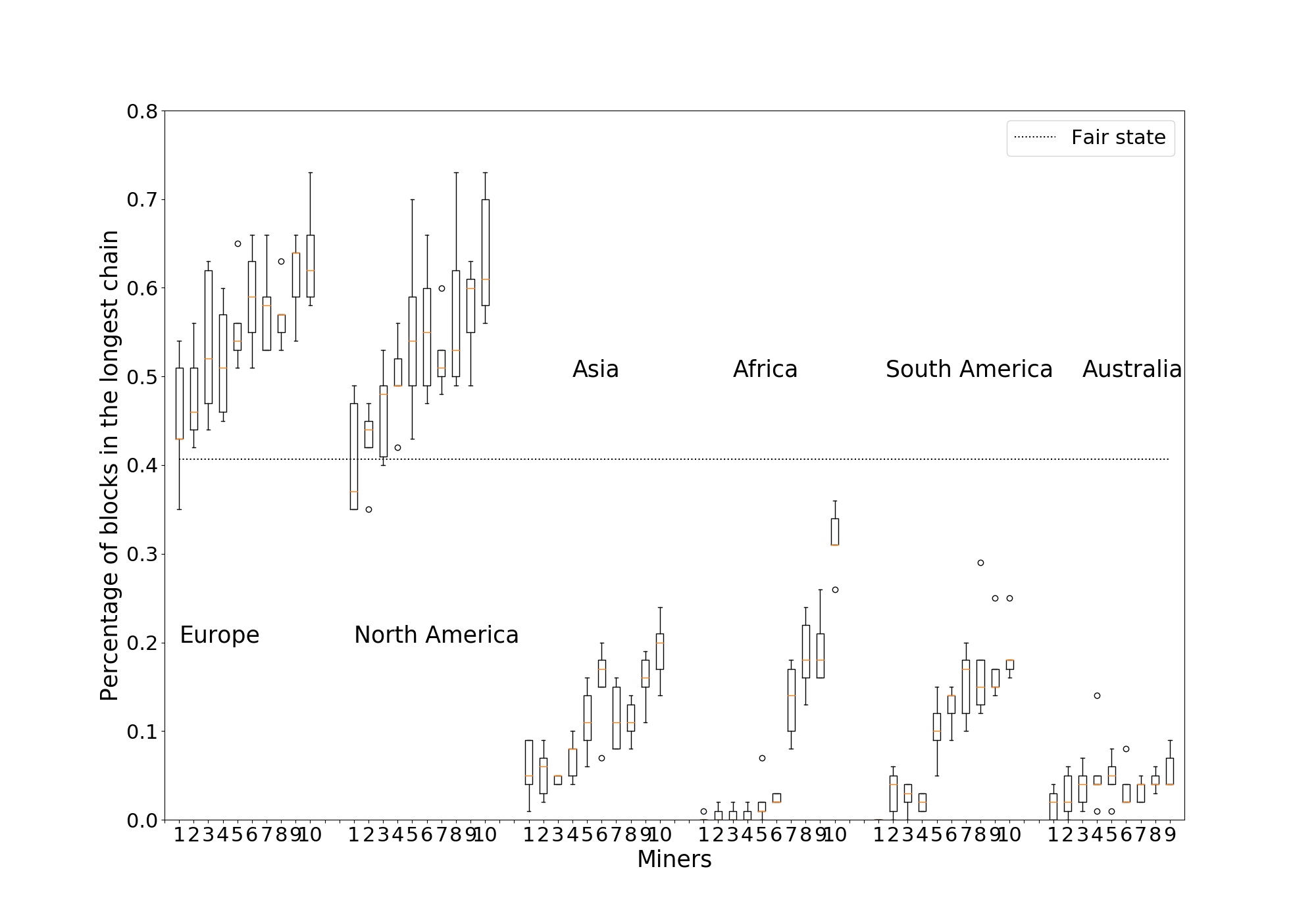} 
\caption{}
\label{fig:EUNACluster}
\end{subfigure}
\caption{Percentage of blocks in longest chain mined by miners in the dominant and non-dominant clusters under (a) random clustering policy, (b) a geography-based clustering policy. EU, NA, AS, AF, SA, AU denote Europe, North America, Asia, Africa, South America and Australia respectively.}
\label{fig:ArtClustering} 
\end{figure}

In light of the reward biases due to clustering that has been highlighted thus far, a natural question that arises is what can miners negatively affected by the clustering do to lessen the impact of clustering on their payoff. 
From the results of \S\ref{s: analysis}, two potential solutions arise. 
\begin{itemize}
\item First, if the miner has knowledge about the existence of a dominant cluster in the network, it should attempt to connect to as many miners within the cluster as possible and become a part of it. 
\item When this is not possible, if the miner has knowledge about other miners that are not included in the dominant cluster it should attempt to connect to as many of those miners as possible to form as large a non-dominant cluster in the network as possible.  
\end{itemize} 
To test these solutions, we conduct experiments where we adjust the neighbors of miners in non-central locations, namely Asia, Africa, Australia and South America to see if their reward improves.
We first focus on miners only in Asia. 
Due to the relatively small number of Asian miners, and relatively large propagation delays to other miners, most Asian nodes get lower than their fair share in the random topology (Figure~\ref{fig:randomandcom}).
By keeping the rest of the network topology fixed (each miner makes 6 random outgoing connections), we adjust each Asian miner to choose 3, 6, 12, 24 and 48 neighbors across different experimental runs. 
Figure~\ref{fig:AsiaIncrease} plots the percentage rewards obtained by 4 randomly selected miners in Asia under this experiment.
With increasing number of neighbor connections, all 4 Asian nodes get a higher reward percentage.
Miners in Russia and Israel are closer to Europe so they perform better than the other 2 nodes at first. 
However, they have a lesser relative performance improvement with 48 outgoing neighbor connections compared with the miners in Japan and India.

Next, we consider the policy where 
we take miners in continents except Europe and North America to form a single tightly-connected non-dominant cluster. 
With the neighbors for Europe and North American miners unchanged, each miner in other continents choose to connect to every other miner in the other continents.  
Figure~\ref{fig:remoteComplete} plots the performance for 10 randomly picked miners from each continent.
Asian and Australian miners get significant benefit compared to the random topology.
African nodes are only slightly affected, with performance rather similar to that of the random topology.
South American miners exhibit two distinct trends: miners in the north of the continent (closer to North America) which perform well under a random topology see some loss; 
while miners in the south of the continent get a performance increase.
European and North American miners lose a bit in their reward percentages.
Thus, the policy of tightly interconnecting miners in non-central locations benefits most of these miners. 
Miners in centrally located clusters or close to the central cluster bear the corresponding loss.

\begin{figure} \begin{subfigure}[t]{0.4\textwidth}
\centering
\includegraphics[width=\textwidth]{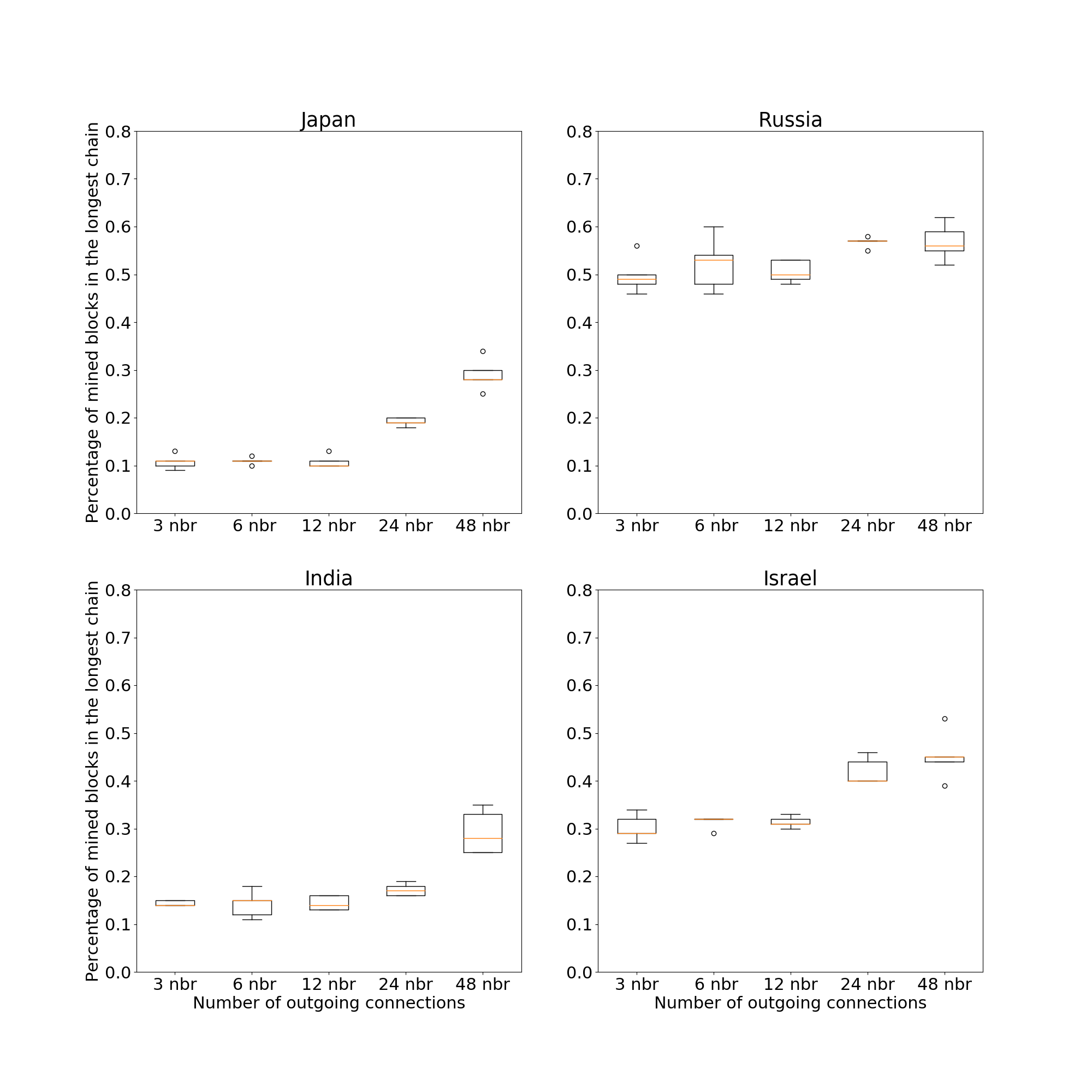} 
\caption{}
\label{fig:AsiaIncrease}
\end{subfigure}
\hfill 
\begin{subfigure}[t]{0.6\textwidth}
\centering
\includegraphics[width=\textwidth]{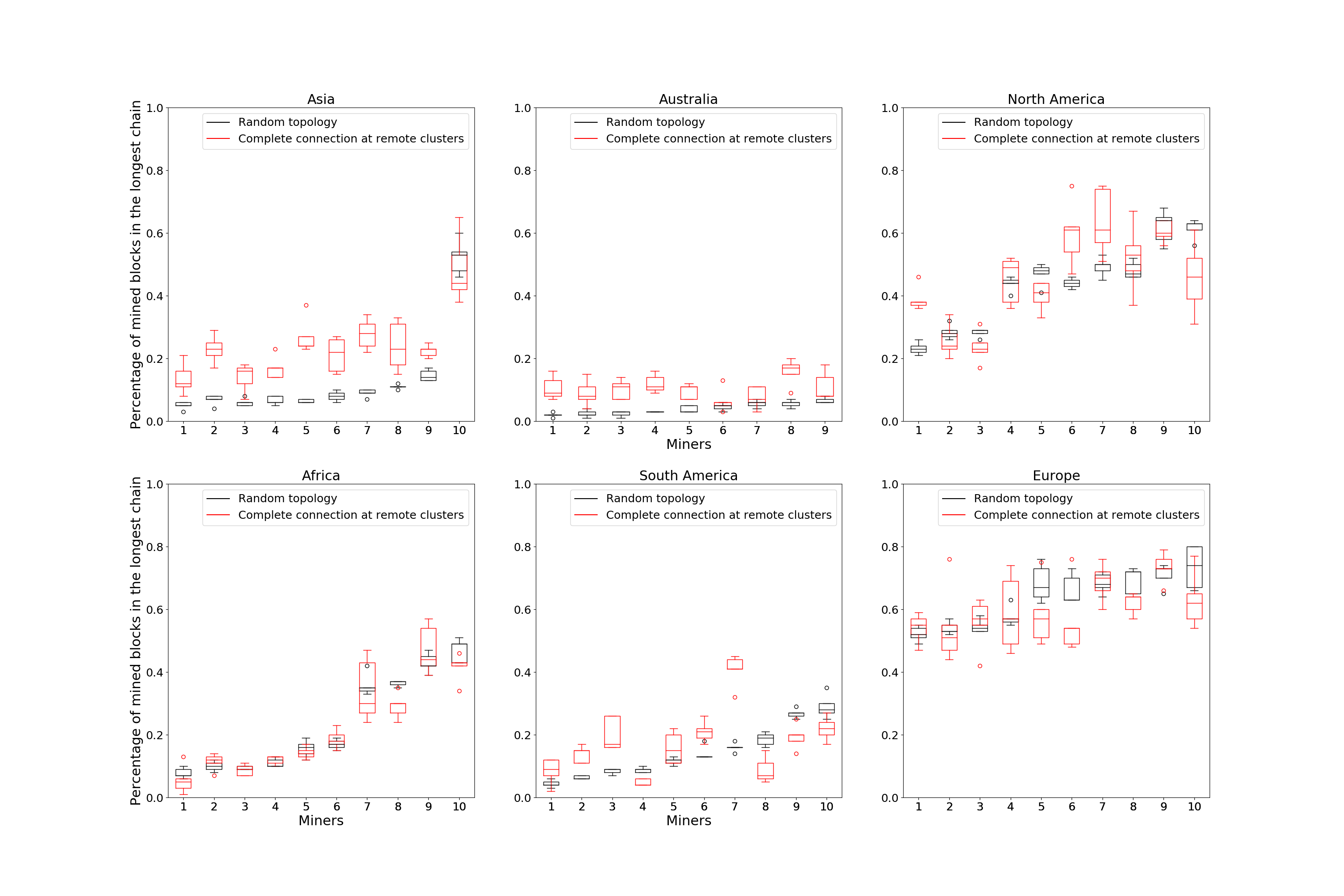} 
\caption{}
\label{fig:remoteComplete}
\end{subfigure}
\caption{(a) Increasing the degree of miners in Asia strictly benefits them. (b) Tightly interconnecting miners in non-central regions strictly benefits most of the non-central miners. }
\label{fig:Asia} 
\end{figure}

Based on the results in Figure~\ref{fig:Asia} and Figure~\ref{fig:EuNA}, the ideal strategy for each miner varies. 
Centrally located miners need to limit their connections, as more neighbors hurt them.
Miners in non-central regions need to expand their neighborhood size, as the more the better.
There is a conflict between these ideal strategies, because when non-central miners increase their number of connections to central miners, it inadvertently also increases the number of connections from the central miners which hurts the central miners.
Therefore central miners are incentivized to reject these extra connections.
The best strategy then for the non-central miners is to connect to the miners that are willing to accept their connection requests.
Ultimately, this creates a situation where non-central miners have a large number of interconnection links among themselves, while central nodes enforce a strict control of their connections.


\section{Related Work}

The effect of network delay on mining rewards in blockchains has been studied in prior works~\cite{shahsavari2019theoretical}.
E.g., Xiao et al.~\cite{xiao2020modeling} show that a slower network leads to high forking rates, and well-connected miners benefit more under heterogeneous network connectivity. 
Cao et al.~\cite{cao2021characterizing} show that with a limited degree bound, high hash power mining pools are incentivized to connect with other hash power mining pools to reduce block propagation delay and increase their mining rewards. 
Gencer et al.~\cite{gencer2018decentralization} highlight that the top miners are generally more successful in including blocks in the main chain, on Bitcoin and Ethereum networks. 
In Putri et al.~\cite{putri2018effect}, the authors look at the effect of network latency on miners executing a selfish-mining attack~\cite{eyal2014majority}.  
These works do not consider the effect increasing connectivity of a miner has on the connectivity of other miners. 


Enabling faster broadcasting of blocks in the network, even if it does not necessarily result in strict reward improvements, allows blocks to be mined faster thereby increasing the transaction confirmation throughput of the system.   
Thus a number of prior works have proposed techniques to accelerate block dissemination. 
Relay networks such as Falcon~\cite{falcon}, Fibre~\cite{fibre}, bloXroute~\cite{klarman2018bloxroute} use a network of dedicated servers to transport blocks quickly between distant geographical regions instead of relying on p2p gossip. 
Kadcast ~\cite{rohrer2019kadcast} proposes a structured p2p overlay with UDP and forward error correction to achieve a more efficient broadcast compared to Bitcoin's random topology with TCP links. 
Perigee ~\cite{mao2020perigee} presents an algorithm to explore and exploit neighbors to search for the best set of neighbors to connect to, for optimizing broadcasting delay. 
Other innovations, e.g., Bitcoin's compact block relay~\cite{corallo2017compact} also help to reduce broadcast time by reducing the size of each block. 
Nagayama et al.~\cite{nagayama2020identifying} study how progress in Internet speeds in recent years have led to reductions in broadcast time in Bitcoin. 

Another line of work considers how miners can deviate from a prescribed mining protocol to selfishly garner benefit~\cite{sapirshtein2016optimal,bag2016bitcoin}. 
Eyal et al.~\cite{eyal2014majority} present a selfish mining attack, wherein a pool of colluding miners withhold mined blocks and reveal them at a later point in time, thus deceiving other miners into performing the futile task of mining blocks over a forked chain. 
Empirical analysis of the Bitcoin blockchain shows a large frequency of short time intervals between consecutive blocks, implying selfish mining may be happening in practice~\cite{neudecker2019short}. 
Lewenberg et al.~\cite{lewenberg2015bitcoin} uses game theoretic tools to analyze how members of a mining pool may share rewards between them. 
They show under certain cases participants are incentivized to keep switching between pools. 

\section{Conclusion}
\label{s: conclusion}

We have considered the problem of topology design for maximizing miner rewards in wide-area PoW blockchain networks.  
A miner may decrease its latency to other miners by carefully choosing its neighbors on the overlay, and/or by increasing the number of neighbors. 
Contrary to general understanding that achieving a lower network latency to other miners is strictly beneficial to miners, we have shown the existence of an ``optimal'' amount of connectivity where a miner receives the most benefit in both noncooperative and cooperative settings. 
Designing fully decentralized algorithms that can achieve this optimal clustering is an interesting direction for future work. 
Generalizing our analysis to arbitrary network topologies is another interesting direction. 
Our result exposes an inherent weakness in PoW blockchains, where miners can always organize themselves in order to selfishly benefit by denying the fair share of rewards to miners of the non-dominant cluster. 
Beyond network-level strategies, it would be useful to propose protocol changes (e.g., at the consensus layer) that incentivize miners to obtain their fair share of rewards.   



\bibliography{paper}

\begin{thebibliography}{10}

\bibitem{avgtxnpblk}
Average transactions per block.
\newblock \url{https://www.blockchain.com/charts/n-transactions-per-block}.

\bibitem{btcavgfee}
Bitcoin average transaction fee.
\newblock \url{https://ycharts.com/indicators/bitcoin_average_transaction_fee}.

\bibitem{bitnodes}
Bitnodes network snapshot.
\newblock \url{https://bitnodes.earn.com/nodes/}.

\bibitem{cryptmarkcap}
Cryptocurrency market capitalization.
\newblock \url{https://coinmarketcap.com/}.

\bibitem{Dogecoin}
Dogecoin.
\newblock \url{https://dogecoin.com/}.

\bibitem{ethereum}
Ethereum.
\newblock \url{https://ethereum.org/en/}.

\bibitem{wondernetwork}
Global ping statistics.
\newblock \url{https://wondernetwork.com/pings}.

\bibitem{omnetppsim}
Omnet++ discrete event simulator.
\newblock \url{https://omnetpp.org/}.

\bibitem{txnfeepbr}
Transaction fees percentage in block reward.
\newblock \url{https://btc.com/stats/fee}.

\bibitem{bitnodescom}
Bitnodes network, 2020.
\newblock Data drawn from website, \url{https://bitnodes.earn.com/}.

\bibitem{falcon}
Falcon, 2020.
\newblock \url{https://www.falcon-net.org/}.

\bibitem{fibre}
Fibre, 2020.
\newblock \url{https://bitcoinfibre.org/}.

\bibitem{abou2019blockchain}
Joe Abou~Jaoude and Raafat~George Saade.
\newblock Blockchain applications--usage in different domains.
\newblock {\em IEEE Access}, 7:45360--45381, 2019.

\bibitem{bag2016bitcoin}
Samiran Bag, Sushmita Ruj, and Kouichi Sakurai.
\newblock Bitcoin block withholding attack: Analysis and mitigation.
\newblock {\em IEEE Transactions on Information Forensics and Security},
  12(8):1967--1978, 2016.

\bibitem{bagaria2019prism}
Vivek Bagaria, Sreeram Kannan, David Tse, Giulia Fanti, and Pramod Viswanath.
\newblock Prism: Deconstructing the blockchain to approach physical limits.
\newblock In {\em Proceedings of the 2019 ACM SIGSAC Conference on Computer and
  Communications Security}, pages 585--602, 2019.

\bibitem{cao2021characterizing}
Tong Cao, J{\'e}r{\'e}mie Decouchant, Jiangshan Yu, and Paulo
  Esteves-Verissimo.
\newblock Characterizing the impact of network delay on bitcoin mining.
\newblock In {\em 2021 40th International Symposium on Reliable Distributed
  Systems (SRDS)}, pages 109--119. IEEE, 2021.

\bibitem{chawla2019velocity}
Nakul Chawla, Hans~Walter Behrens, Darren Tapp, Dragan Boscovic, and
  K~Sel{\c{c}}uk Candan.
\newblock Velocity: Scalability improvements in block propagation through
  rateless erasure coding.
\newblock In {\em 2019 IEEE International Conference on Blockchain and
  Cryptocurrency (ICBC)}, pages 447--454. IEEE, 2019.

\bibitem{corallo2017compact}
Matt Corallo.
\newblock Compact block relay. bip 152, 2017.

\bibitem{croman2016scaling}
Kyle Croman, Christian Decker, Ittay Eyal, Adem~Efe Gencer, Ari Juels, Ahmed
  Kosba, Andrew Miller, Prateek Saxena, Elaine Shi, Emin~G{\"u}n Sirer, et~al.
\newblock On scaling decentralized blockchains.
\newblock In {\em International conference on financial cryptography and data
  security}, pages 106--125. Springer, 2016.

\bibitem{decker2013information}
Christian Decker and Roger Wattenhofer.
\newblock Information propagation in the bitcoin network.
\newblock In {\em IEEE P2P 2013 Proceedings}, pages 1--10. IEEE, 2013.

\bibitem{eyal2014majority}
Ittay Eyal and Emin~G{\"u}n Sirer.
\newblock Majority is not enough: Bitcoin mining is vulnerable.
\newblock In {\em International conference on financial cryptography and data
  security}, pages 436--454. Springer, 2014.

\bibitem{gencer2018decentralization}
Adem~Efe Gencer, Soumya Basu, Ittay Eyal, Robbert Van~Renesse, and Emin~G{\"u}n
  Sirer.
\newblock Decentralization in bitcoin and ethereum networks.
\newblock In {\em International Conference on Financial Cryptography and Data
  Security}, pages 439--457. Springer, 2018.

\bibitem{gervais2015tampering}
Arthur Gervais, Hubert Ritzdorf, Ghassan~O Karame, and Srdjan Capkun.
\newblock Tampering with the delivery of blocks and transactions in bitcoin.
\newblock In {\em Proceedings of the 22nd ACM SIGSAC Conference on Computer and
  Communications Security}, pages 692--705, 2015.

\bibitem{Gilad2017algorand}
Yossi Gilad, Rotem Hemo, Silvio Micali, Georgios Vlachos, and Nickolai
  Zeldovich.
\newblock Algorand: Scaling byzantine agreements for cryptocurrencies.
\newblock In {\em Proceedings of the $26$th Symposium on Operating Systems
  Principle}, pages 51--68, 2017.

\bibitem{hari2019accel}
Adiseshu Hari, Murali Kodialam, and TV~Lakshman.
\newblock Accel: Accelerating the bitcoin blockchain for high-throughput,
  low-latency applications.
\newblock In {\em IEEE infocom 2019-IEEE conference on computer
  communications}, pages 2368--2376. IEEE, 2019.

\bibitem{kiffer2018better}
Lucianna Kiffer, Rajmohan Rajaraman, and Abhi Shelat.
\newblock A better method to analyze blockchain consistency.
\newblock In {\em Proceedings of the 2018 ACM SIGSAC Conference on Computer and
  Communications Security}, pages 729--744, 2018.

\bibitem{klarman2018bloxroute}
Uri Klarman, Soumya Basu, Aleksandar Kuzmanovic, and Emin~G{\"u}n Sirer.
\newblock bloxroute: A scalable trustless blockchain distribution network
  whitepaper.
\newblock {\em IEEE Internet of Things Journal}, 2018.

\bibitem{kroll2013economics}
Joshua~A Kroll, Ian~C Davey, and Edward~W Felten.
\newblock The economics of bitcoin mining, or bitcoin in the presence of
  adversaries.
\newblock In {\em Proceedings of WEIS}, volume 2013, page~11, 2013.

\bibitem{lewenberg2015bitcoin}
Yoad Lewenberg, Yoram Bachrach, Yonatan Sompolinsky, Aviv Zohar, and Jeffrey~S
  Rosenschein.
\newblock Bitcoin mining pools: A cooperative game theoretic analysis.
\newblock In {\em Proceedings of the 2015 international conference on
  autonomous agents and multiagent systems}, pages 919--927. Citeseer, 2015.

\bibitem{mao2020perigee}
Yifan Mao, Soubhik Deb, Shaileshh~Bojja Venkatakrishnan, Sreeram Kannan, and
  Kannan Srinivasan.
\newblock Perigee: Efficient peer-to-peer network design for blockchains.
\newblock In {\em Proceedings of the 39th Symposium on Principles of
  Distributed Computing}, pages 428--437, 2020.

\bibitem{marcus2018low}
Yuval Marcus, Ethan Heilman, and Sharon Goldberg.
\newblock Low-resource eclipse attacks on ethereum's peer-to-peer network.
\newblock {\em IACR Cryptology ePrint Archive}, 2018(236), 2018.

\bibitem{miller2015discovering}
Andrew Miller, James Litton, Andrew Pachulski, Neal Gupta, Dave Levin, Neil
  Spring, and Bobby Bhattacharjee.
\newblock Discovering bitcoin?s public topology and influential nodes.
\newblock {\em et al}, 2015.

\bibitem{mukhopadhyay2016brief}
Ujan Mukhopadhyay, Anthony Skjellum, Oluwakemi Hambolu, Jon Oakley, Lu~Yu, and
  Richard Brooks.
\newblock A brief survey of cryptocurrency systems.
\newblock In {\em 2016 14th annual conference on privacy, security and trust
  (PST)}, pages 745--752. IEEE, 2016.

\bibitem{nagayama2020identifying}
Ryunosuke Nagayama, Ryohei Banno, and Kazuyuki Shudo.
\newblock Identifying impacts of protocol and internet development on the
  bitcoin network.
\newblock In {\em 2020 IEEE Symposium on Computers and Communications (ISCC)},
  pages 1--6. IEEE, 2020.

\bibitem{genesis}
Satoshi Nakamoto.
\newblock A peer-to-peer electronic cash system, 2008.
\newblock \url{http: //bitcoin.org/bitcoin}.

\bibitem{neudecker2019short}
Till Neudecker and Hannes Hartenstein.
\newblock Short paper: An empirical analysis of blockchain forks in bitcoin.
\newblock In {\em International Conference on Financial Cryptography and Data
  Security}, pages 84--92. Springer, 2019.

\bibitem{ozisik2017graphene}
A~Pinar Ozisik, Gavin Andresen, George Bissias, Amir Houmansadr, and Brian
  Levine.
\newblock Graphene: A new protocol for block propagation using set
  reconciliation.
\newblock In {\em Data Privacy Management, Cryptocurrencies and Blockchain
  Technology}, pages 420--428. Springer, 2017.

\bibitem{park2019nodes}
Sehyun Park, Seongwon Im, Youhwan Seol, and Jeongyeup Paek.
\newblock Nodes in the bitcoin network: Comparative measurement study and
  survey.
\newblock {\em IEEE Access}, 7:57009--57022, 2019.

\bibitem{putri2018effect}
Bellia Dwi~Cahya Putri and Riri~Fitri Sari.
\newblock The effect of latency on selfish-miner attack on block receive time
  bitcoin network using ns3.
\newblock In {\em 2018 12th International Conference on Telecommunication
  Systems, Services, and Applications (TSSA)}, pages 1--5. IEEE, 2018.

\bibitem{ritz2018impact}
Fabian Ritz and Alf Zugenmaier.
\newblock The impact of uncle rewards on selfish mining in ethereum.
\newblock In {\em 2018 IEEE European Symposium on Security and Privacy
  Workshops (EuroS\&PW)}, pages 50--57. IEEE, 2018.

\bibitem{rohrer2019kadcast}
Elias Rohrer and Florian Tschorsch.
\newblock Kadcast: A structured approach to broadcast in blockchain networks.
\newblock In {\em Proceedings of the 1st ACM Conference on Advances in
  Financial Technologies}, pages 199--213, 2019.

\bibitem{rosenfeld2014analysis}
Meni Rosenfeld.
\newblock Analysis of hashrate-based double spending.
\newblock {\em arXiv preprint arXiv:1402.2009}, 2014.

\bibitem{sapirshtein2016optimal}
Ayelet Sapirshtein, Yonatan Sompolinsky, and Aviv Zohar.
\newblock Optimal selfish mining strategies in bitcoin.
\newblock In {\em International Conference on Financial Cryptography and Data
  Security}, pages 515--532. Springer, 2016.

\bibitem{schaffer2019performance}
Markus Sch{\"a}ffer, Monika Di~Angelo, and Gernot Salzer.
\newblock Performance and scalability of private ethereum blockchains.
\newblock In {\em International Conference on Business Process Management},
  pages 103--118. Springer, 2019.

\bibitem{shahsavari2019theoretical}
Yahya Shahsavari, Kaiwen Zhang, and Chamseddine Talhi.
\newblock A theoretical model for fork analysis in the bitcoin network.
\newblock In {\em 2019 IEEE International Conference on Blockchain
  (Blockchain)}, pages 237--244. IEEE, 2019.

\bibitem{thum2018economic}
Marcel Thum.
\newblock The economic cost of bitcoin mining.
\newblock In {\em CESifo Forum}, volume~19, pages 43--45. M{\"u}nchen: ifo
  Institut-Leibniz-Institut f{\"u}r Wirtschaftsforschung an der~…, 2018.

\bibitem{varga2008overview}
Andr{\'a}s Varga and Rudolf Hornig.
\newblock An overview of the omnet++ simulation environment.
\newblock In {\em Proceedings of the 1st international conference on Simulation
  tools and techniques for communications, networks and systems \& workshops},
  pages 1--10, 2008.

\bibitem{wan2019evaluating}
Luming Wan, David Eyers, and Haibo Zhang.
\newblock Evaluating the impact of network latency on the safety of blockchain
  transactions.
\newblock In {\em 2019 IEEE International Conference on Blockchain
  (Blockchain)}, pages 194--201. IEEE, 2019.

\bibitem{wang2021ethna}
Taotao Wang, Chonghe Zhao, Qing Yang, Shengli Zhang, and Soung~Chang Liew.
\newblock Ethna: Analyzing the underlying peer-to-peer network of ethereum
  blockchain.
\newblock {\em IEEE Transactions on Network Science and Engineering},
  8(3):2131--2146, 2021.

\bibitem{xiao2020modeling}
Yang Xiao, Ning Zhang, Wenjing Lou, and Y~Thomas Hou.
\newblock Modeling the impact of network connectivity on consensus security of
  proof-of-work blockchain.
\newblock In {\em IEEE INFOCOM 2020-IEEE Conference on Computer
  Communications}, pages 1648--1657. IEEE, 2020.

\end{thebibliography}

\appendix
\section{Proof of Lemma~\ref{lem:twoclus}}
\label{apx:lemmatwoclusproof}
\begin{proof}
Consider a fork phase sequence $c_1, c_2, \ldots, c_{2i}$, where $c_j$ for  $j\in \{1, \ldots, 2i\}$ denotes the cluster that mined the $j$-th block in the fork phase sequence. 
Let $b_1, b_2, \ldots, b_{2i}$ be the sequence of blocks that are mined during this phase.
Supposing the fork phase started right at round 1, meaning the first block $b_1$ mined by $c_1$ is mined on top of the genesis block. 
At the beginning of the second round, none of the miners in cluster $c_2$ would have received $b_1$.
Therefore $b_2$ is mined on top of the genesis block as well. 
Thus there are two forks at the end of two rounds, one due to each cluster. 
Similarly during the third round, if $c_3 = c_1$ then block $b_3$ is going to be mined on top of $b_1$, whereas if $c_3 = c_2$ then block $b_3$ is going to be mined on top of $b_2$.
This is because if $c_3 = c_1$ then the miner of $b_3$ would have received $b_1$ but not $b_2$ at the beginning of the third round. 
Since the fork ending with $b_1$ is the longest fork known to the miner of $b_3$, it mines the block $b_3$ on top of $b_1$. 
Analogously, if $c_3 = c_2$ then the miner of $b_3$ receives $b_2$ before $b_1$. 
Since both forks are at equal height but $b_2$ was received earlier than $b_1$, the miner of block $b_3$ mines $b_3$ on top of $b_2$.  
We can similarly argue that if $c_4 = c_2$ then $b_4$ is mined on top of $b_2$, and if $c_4 = c_1$ then $b_4$ is mined on top of $b_1$. 
Continuing this reasoning, by induction we can see that throughout the duration of the fork phase there are two forks (one per cluster) that increases in height in parallel. 

If the fork phase did not start right at round 1, let $c_0$ be the cluster that mined the block $b_0$ in the round before the fork phase started and let $b_{-1}$ be the block mined in the round before $b_0$ was mined (it is possible $b_{-1}$ is the genesis block). 
Clearly $c_0 \neq c_1$, since this would mean $c_0$ is also in the fork phase. 
Therefore $c_0 = c_1$ and block $b_1$ is mined on top of block $b_0$. 
Moreover, the height of $b_{-1}$ is strictly less than the height of $b_0$. 
For otherwise $b_{-1}, b_0$ would also be in the fork phase. 
Therefore $b_2$ is also mined on top of $b_0$. 
From here on, we continue the same argument as in the case where the fork phase starts at round 1 to show the lemma. 
\end{proof}

\section{Proof of Theorem~\ref{thm:twocluster}}
\label{apx:thmtwoclusproof}

\begin{proof}
We consider three possible cases: \\
{\em (i) $P_1$ is a 1-run phase.} Supposing the phase lasts for $i \geq 1$ rounds. 
The probability of this happening is $p^{i+1} (1-p)$, since the phase ends only when the $(i+1)$-th block is mined by cluster 1 and the $(i+2)$-th block is mined by cluster 2.
All blocks mined by cluster 1 during this phase belong to the longest chain. \\
{\em (ii) $P_1$ is a 2-run phase.} Again suppose the phase lasts for $i \geq 1$ rounds. 
The probability of this event happening is $(1-p)^{i+1} p$ by the same argument as before. 
None of the blocks mined during this phase belong to cluster 1. \\
{\em (iii) $P_1$ is a fork phase.} Now, suppose the phase lasts for $2i$ rounds for $i\geq 0$. 
Further suppose that cluster 1 wins this phase, i.e., blocks mined by cluster 1 become part of the longest chain after the phase ends. 
The probability of this occurring is $(2p(1-p))^i p^2$. 
Here $2p(1-p)$ is the probability that in every consecutive pair of rounds, one of the blocks is mined by one cluster and the other block is mined by the other cluster. 
$p^2$ is the probability that at the end of the fork, two blocks are mined by cluster 1 which ensures that cluster 1's fork wins out after the phase ends. 
In this case all $i$ blocks mined by cluster 1 during the phase become part of the longest chain. 
Similarly $(2p(1-p))^i (1-p)^2$ is the probability that cluster 2 wins out at the end of the phase. 
However, in this case none of the blocks mined by cluster 1 become part of the longest chain. 

From the above three cases, we have
\begin{align}
\mathbb{E}[P_1^1] &= \sum_{i=1}^\infty \left( i p^{i+1} (1-p) + i (2p(1-p))^i p^2 \right) \notag \\
&= \frac{p^2}{1-p} + \frac{2p^3(1-p)}{(1-2p(1-p))^2} \notag \\
\Rightarrow \mathbb{E}[M_1] &= k \left[ \frac{p^2}{1-p} + \frac{2p^3(1-p)}{(1-2p(1-p))^2} \right].   
\end{align}
We can similarly compute the expected number of blocks mined by cluster 2 as 
\begin{align}
\mathbb{E}[M_2] = k \left[ \frac{(1-p)^2}{p} + \frac{2(1-p)^3p}{(1-2p(1-p))^2} \right]. 
\end{align}
Since each miner in cluster 1 is equally likely to mine a block, by symmetry we have 
\begin{align}
\mathbb{E}[M_v] = \frac{k}{n} \left[ \frac{p}{1-p} + \frac{2p^2(1-p)}{(1-2p(1-p))^2} \right] \quad \forall v \in V_1. \label{eq:expfracpermin2}
\end{align}
In the limit of the number of phases approaching infinity, using concentration bounds we can show that the expected percentage of blocks mined by a miner in cluster 1 can be given as 
\begin{align}
\mathbb{E}[F_v] = \frac{\mathbb{E}[M_v]}{\mathbb{E}[M_1] + \mathbb{E}[M_2]} 
\end{align}
where $\mathbb{E}[M_1], \mathbb{E}[M_2], \mathbb{E}[M_v]$ are as above. 
\end{proof}

\section{Proof of Theorem~\ref{thm:fprimetwoclus}}
\label{apx:fprimetwoclusproof}
\begin{proof}
The proof follows a similar argument as in the proof of Theorem~\ref{thm:twocluster}.  
As before let $P_1, P_2, \ldots, P_k$ be a sequence of $k$-phases starting from the genesis block.  
The expected number of blocks mined by a miner $v\in V_1$ that are included within the longest chain during these $k$ phases is given by Equation~\eqref{eq:expfracpermin2}. 
Let $\tilde{P}_1^c$ be the expected total number of blocks mined by cluster $c$ in phase 1.
During a 1-run phase all blocks mined belong to cluster 1. 
During a 2-run phase none of the blocks mined belong to cluster 1. 
During a fork phase, half of all blocks mined belong to cluster 1. 
Therefore, 
\begin{align}
\mathbb{E}[\tilde{P}_1^1] &= \sum_{i=1}^\infty \left( i p^{i+1} (1-p) + i (2p(1-p))^i p^2 + i (2p(1-p))^i (1-p)^2 \right) \\
&= \frac{p^2}{1-p} + \frac{2p^3(1-p)}{(1-2p(1-p))^2} +  \frac{2p(1-p)^3}{(1-2p(1-p))^2}. 
\end{align} 
Hence the expected total number of blocks mined by a miner $v \in V_1$ across $k$ phases is 
\begin{align}
\mathbb{E}[\tilde{M}_v] = \frac{k}{n} \left[ \frac{p}{1-p} + \frac{2p^2(1-p)}{(1-2p(1-p))^2} +  \frac{2(1-p)^3}{(1-2p(1-p))^2} \right]. \label{eq:twoclustotnum}
\end{align}
Combining Equation~\eqref{eq:twoclustotnum} with Equation~\eqref{eq:expfracpermin2} and using the definition of $W_v$ in Equation~\eqref{eq:Wdefn}, the Theorem follows. 
\end{proof}

\section{Optimal Cluster Size Depends on Underlying Network Latencies}
\label{apx:optclusdelta}

\begin{figure}[t]
    \centering
    \includegraphics[width=0.6\textwidth]{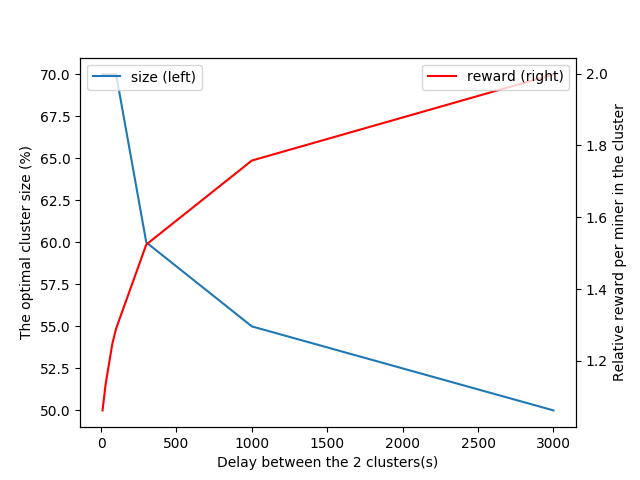}
    \caption{Optimum cluster size and relative gain in reward for the dominant cluster as the inter-cluster latency $\Delta$ varies. }
    \label{fig:2clusterDelay}
\end{figure}
To understand the impact of network latencies on optimum cluster size, we consider the two cluster setting of \S\ref{sec: two clusters} and simulate it on OMNet++ according to the model discussed in \S\ref{s: sysmodel}. 
The network has a total of 60 nodes and the intra-cluster latency $\epsilon = 1$. 
The mining rate for each miner is chosen to be $h_v = 1/2400$, so that one block is mined on average once every 40 time units. 
We vary the inter-cluster latency $\Delta$ as $10, 30, 75, 100, 300, 1000, 3000$. 
For each candidate $\Delta$, we consider different possible sizes for the dominant cluster ($55, 60, 65, 70, 75, 80, 85, 90$) and evaluate the reward gains for each size. 
Figure~\ref{fig:2clusterDelay} plots the optimal observed cluster size and the corresponding reward for the dominant cluster miners, as $\Delta$ varies. 
We observe that as $\Delta$ increases the size of the optimal cluster drops, from about 70\% when $\Delta = 10$ to about 50\% when $\Delta=3000$. 
The percentage increase in reward increases from 6\% at $\Delta=10$ to about 100\% when $\Delta=3000$. 

\section{Three Clusters}

\begin{figure}[t]
    \centering
    \includegraphics[width=0.98\textwidth]{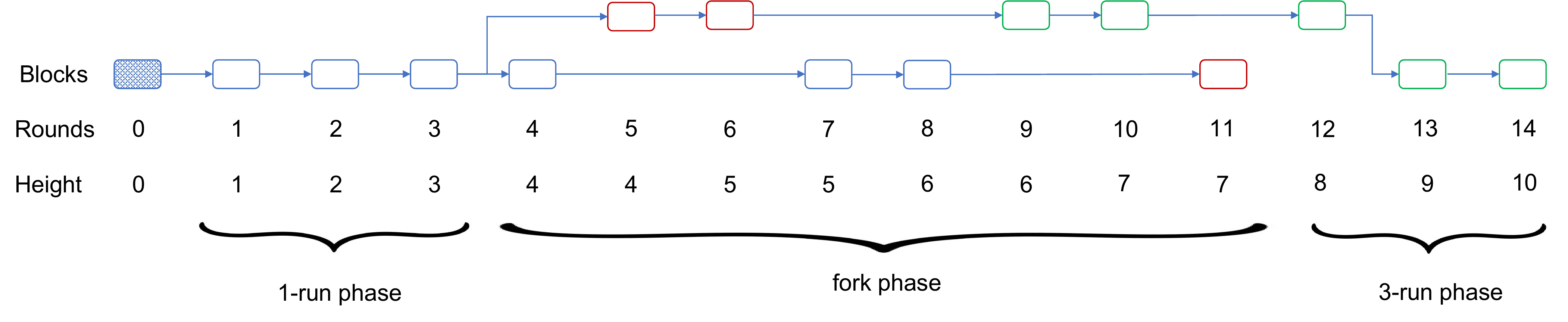}
    \caption{Evolution of the blockchain when there are three clusters. Blocks mined by each cluster are denoted using a different color. The figure also illustrates the different phases that occur during the evolution. }
    \label{fig:threecluster_blockdiagram}
\end{figure}

Next, we consider the case when there are three clusters in the network $V_1, V_2, V_3$ with $|V_1| = p_1 n, |V_2| = p_2 n, |V_3| = p_3 n$ for $p_1 \geq 0, p_2 \geq 0, p_3 \geq 0$ with $p_1 + p_2 + p_3 = 1$. 
We assume the latency between two miners within the same cluster is $\epsilon$ with $0 < \epsilon < 1$ while the latency between two miners in different clusters is $\Delta$ with $1 + \epsilon < \Delta < 2$.  
Consider an example where the cluster to mine a block during each round is given by the sequence $1, 1, 1, 1, 2, 2, 1, 1, 3, 3, 2, 3, 3, 3$ as shown in Figure~\ref{fig:threecluster_blockdiagram}. 
As in the two cluster case, we observe distinct phases that occur in the sequence. 
During rounds 1 to 3 we have a 1-run phase, where blocks are mined by cluster 1 and all blocks belong to the longest chain. 
During rounds 4 to 11 we have a fork phase, where there are two competing chains that are growing in parallel. 
From round 12 onwards we have a 3-run phase where blocks are mined by cluster 2. 

Notice that the number of forks occurring during a fork phase is still two, as in the two cluster case. 
This is due to our assumption that $\Delta < 2$, which guarantees each miner to hear about at least the last but one block at the beginning of a round. 
While the 1-run, 2-run and 3-run phases are analogous to the two cluster case, the fork phase in this case behaves differently than before. 
Specifically, when a fork phase ends it is no longer true that a single cluster emerges as the winner of the phase. 
Indeed, in Figure~\ref{fig:threecluster_blockdiagram} we see that the winning fork includes blocks mined by all three clusters. 
This makes the analysis more involved compared to the two cluster case. 
\begin{thm}
\label{thm:threeclusterf}
For a three cluster network with latency graph as described in this section, we have 
\begin{align}
\mathbb{E}[F_v] = \frac{\mathbb{E}[M_1]}{p_1(\mathbb{E}[M_1] + \mathbb{E}[M_2] + \mathbb{E}[M_3])},
\end{align}
for any $v \in V_1$ where $\mathbb{E}[M_1], \mathbb{E}[M_2], \mathbb{E}[M_3]$ are as given in Equation~\eqref{eq:threeclumexp} in the proof below.
\end{thm}
\begin{proof}
As in the two cluster case, let $P_1, P_2, \ldots, P_k$ be $k$ phases starting from the genesis block. 
It suffices to analyze the expected number of blocks mined by cluster 1 during the first phase as the phases are identically distributed.  
Let $M_1$ be the number of blocks mined by cluster 1 that is included in the longest chain in phase $P_1$. 
If $P_1$ is a 1-run phase, then all blocks mined during the phase belong the longest chain. 
The contribution of this term in the overall expression for $\mathbb{E}[M_1]$ is 
\begin{align}
\sum_{i=1}^\infty p_1^i p_1 (1-p_1) i = \frac{p_1^2}{1-p_1}, \label{eq: threeclusterinterm}
\end{align} 
by the same reasoning as in the two cluster case. 

Next, suppose $P_1$ is a fork phase. 
Let the blocks mined during this phase be from clusters $c_1, c_2, \ldots, c_{2l-1}, c_{2l}$ where $c_i$ denotes the cluster that mined block $i$ in the phase.
Let $c_{2l+1}, c_{2l+2}$ be the clusters that mined blocks $2l+1, 2l+2$ after the fork phase ends. 
We must have $c_{2l+1} = c_{2l+2}$. 
We observe that the winner among blocks $2l-1, 2l$ is determined by $c_{2l+1}, c_{2l+2}$. 
If $c_{2l-1} = c_{2l+1}$, then $c_{2l-1}$ is the winner; else if $c_{2l} = c_{2l+1}$ then $c_{2l}$ is the winner; if $c_{2l-1} \neq c_{2l} \neq c_{2l+1}$ then $c_{2l-1}$ is the winner. 
Similarly, the winning cluster during round $2l-3, 2l-2$ is determined by the winning cluster during rounds $2l-1, 2l$. 
Let $c^*_{2l-1, 2l}$ be the winning cluster during rounds $2l-1, 2l$. 
If $c_{2l-3} = c^*_{2l-1, 2l}$ then $c_{2l-3}$ is the winner; else if $c_{2l-2} = c^*_{2l-1,2l}$ then $c_{2l-2}$ is the winner; if $c_{2l-3} \neq c_{2l-2} \neq c^*_{2l-1,2l}$ then $c_{2l-3}$ is the winner.
By continuing this process we can determine the winning cluster during each consecutive pair of rounds in the fork phase.  

\begin{figure}[t]
    \centering
    \includegraphics[width=0.98\textwidth]{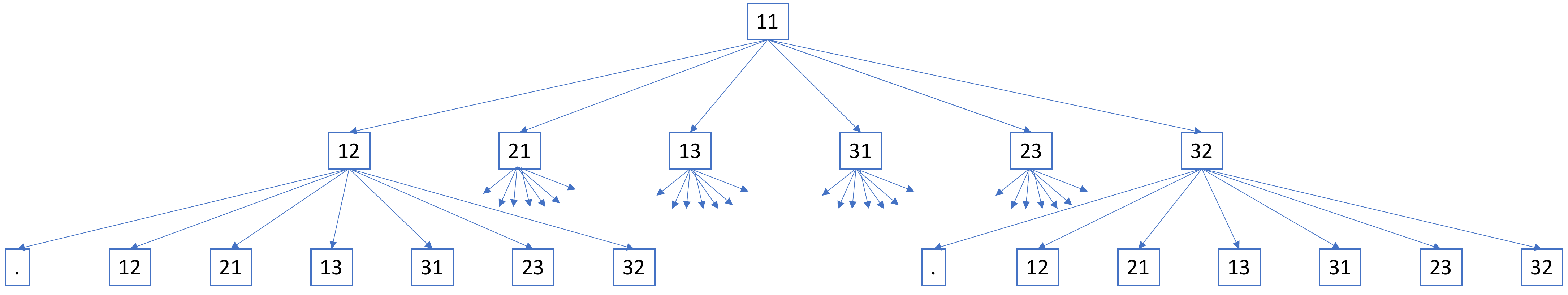}
    \caption{Possible sample paths that can occur during a fork phase when there are three clusters.}
    \label{fig:threeclustertree}
\end{figure}

To analyze the contribution of the fork phase in the overall expression for $\mathbb{E}[M_1]$, consider Figure~\ref{fig:threeclustertree} which shows the different possible sample paths in the fork phase that ends with two blocks mined by cluster 1. 
A sample path starts with the leaf node, and goes up the tree until it reaches the root. 
Each sample path has a certain probability of occurring and has a certain number of blocks mined by cluster 1 that is included in the longest chain. 
For a node $(i, j)$ in Figure~\ref{fig:threeclustertree} in which cluster $i$ wins, let $\alpha_{ij}^i$ be the average number of blocks mined by cluster 1 that are included in the longest chain in the subtree rooted at that node. 
Similarly, for a node $(i, j)$ in Figure~\ref{fig:threeclustertree} in which cluster $j$ wins, let $\alpha_{ij}^j$ be the average number of blocks mined by cluster 1 that are included in the longest chain in the subtree rooted at that node. 
By using the fact that winner of node is dictated by the winner of its parent in the tree, we have following recurrence relations: 
\begin{align}
\alpha^1_{12} = 1 + p_1p_2 \alpha^1_{12} + p_2 p_1 \alpha^1_{21} + p_1 p_3 \alpha^1_{13} + p_3 p_1 \alpha^1_{31} + p_2 p_3 \alpha^2_{2 3} + p_3 p_2 \alpha^3_{32} \\
\alpha^2_{12} = p_1 p_2 \alpha^2_{12} + p_2 p_1 \alpha^2_{21} + p_1p_3\alpha^1_{13} + p_3p_1\alpha^3_{31} + p_2p_3\alpha^2_{23} + p_3 p_2 \alpha^2_{32} \\
\alpha^1_{13} = 1 + p_1p_2\alpha^1_{12} + p_2p_1\alpha^1_{21} + p_1p_3\alpha^1_{13} + p_3p_1\alpha^1_{31} + p_2p_3\alpha^2_{23} + p_3p_2\alpha^3_{32} \\
\alpha^3_{13} = p_1 p_2 \alpha^1_{12} + p_2 p_1 \alpha^2_{21} + p_1 p_3 \alpha^3_{13} + p_3 p_1 \alpha^3_{31} + p_2 p_3 \alpha^3_{23} + p_3 p_2 \alpha^3_{32} \\
\alpha^2_{23} = p_1 p_2 \alpha^2_{12} + p_2 p_1 \alpha^2_{21} + p_1 p_3 \alpha^1_{13} + p_3 p_1 \alpha^3_{31} + p_2 p_3 \alpha^2_{23} + p_3 p_2 \alpha^2_{32} \\
\alpha^3_{23} = p_1 p_2 \alpha^1_{12} + p_2 p_1 \alpha^2_{21} + p_1 p_3 \alpha^3_{13} + p_3 p_1 \alpha^3_{31} + p_2 p_3 \alpha^3_{23} + p_3 p_2 \alpha^3_{32} \\
\alpha^2_{21} = p_1 p_2 \alpha^2_{12} + p_2 p_1 \alpha^2_{21} + p_1 p_3\alpha^1_{13} + p_3 p_1 \alpha^3_{31} + p_2 p_3 \alpha^2_{23} + p_3 p_2 \alpha^2_{32} \\
\alpha^1_{21} = 1 + p_1 p_2 \alpha^1_{12} + p_2 p_1 \alpha^1_{21} + p_1 p_3 \alpha^1_{13} + p_3 p_1 \alpha^1_{31} + p_2 p_3 \alpha^2_{23} + p_3 p_2 \alpha^3_{32} \\
\alpha^3_{31} = p_1 p_2 \alpha^1_{12} + p_2 p_1 \alpha^2_{21} + p_1 p_3 \alpha^3_{13} + p_3 p_1 \alpha^3_{31} + p_2 p_3 \alpha^3_{23} + p_3 p_2 \alpha^3_{32} \\
\alpha^1_{31} = 1 + p_1 p_2 \alpha^1_{12} + p_2 p_1 \alpha^1_{21} + p_1 p_3 \alpha^1_{13} + p_3 p_1 \alpha^1_{31} + p_2 p_3 \alpha^2_{23} + p_3 p_2 \alpha^3_{32} \\
\alpha^3_{32} = p_1 p_2 \alpha^1_{12} + p_2 p_1 \alpha^2_{21} + p_1 p_3 \alpha^3_{13} + p_3 p_1 \alpha^3_{31} + p_2 p_3 \alpha^3_{23} + p_3 p_2 \alpha^3_{32} \\
\alpha^2_{32} = p_1 p_2 \alpha^2_{12} + p_2 p_1 \alpha^2_{21} + p_1 p_3 \alpha^1_{13} + p_3 p_1 \alpha^3_{31} + p_2 p_3 \alpha^2_{23} + p_3 p_2 \alpha^2_{32}. 
\end{align}
Once we solve for all the $\alpha$'s in the above system of equations, we can compute the contribution of the fork phase term in $\mathbb{E}[M_1]$ as 
\begin{align}
p_1^2 ( p_1 p_2 \alpha^1_{12} + p_2 p_1 \alpha^1_{21} + p_1 p_3 \alpha^1_{13} + p_3 p_1 \alpha^1_{31} + p_2 p_3 \alpha^2_{23} + p_3 p_2 \alpha^3_{32}) + \notag \\
 p_2^2(p_1 p_2 \alpha^2_{12} + p_2 p_1 \alpha^2_{21} + p_1 p_3 \alpha^1_{13} + p_3 p_1 \alpha^3_{31} + p_2 p_3 \alpha^2_{23} + p_3 p_2 \alpha^2_{32}) + \notag \\
p_3^2(p_1 p_2 \alpha^1_{12} + p_2 p_1 \alpha^2_{21} + p_1 p_3 \alpha^3_{13} + p_3 p_1 \alpha^3_{31} + p_2 p_3 \alpha^3_{23} + p_3 p_2 \alpha^3_{32}), 
\end{align}
as a fork phase can end with two blocks that are either both mined by cluster 1, both mined by cluster 2 or both mined by cluster 3. 
Combining this with Equation~\eqref{eq: threeclusterinterm} we have 
\begin{align}
\mathbb{E}[M_1] = p_1^2 ( p_1 p_2 \alpha^1_{12} + p_2 p_1 \alpha^1_{21} + p_1 p_3 \alpha^1_{13} + p_3 p_1 \alpha^1_{31} + p_2 p_3 \alpha^2_{23} + p_3 p_2 \alpha^3_{32}) + \notag \\
 p_2^2(p_1 p_2 \alpha^2_{12} + p_2 p_1 \alpha^2_{21} + p_1 p_3 \alpha^1_{13} + p_3 p_1 \alpha^3_{31} + p_2 p_3 \alpha^2_{23} + p_3 p_2 \alpha^2_{32}) + \notag \\
p_3^2(p_1 p_2 \alpha^1_{12} + p_2 p_1 \alpha^2_{21} + p_1 p_3 \alpha^3_{13} + p_3 p_1 \alpha^3_{31} + p_2 p_3 \alpha^3_{23} + p_3 p_2 \alpha^3_{32}) + \notag \\
\sum_{i=1}^\infty p_1^i p_1 (1-p_1) i \\
= p_1^2 ( p_1 p_2 \alpha^1_{12} + p_2 p_1 \alpha^1_{21} + p_1 p_3 \alpha^1_{13} + p_3 p_1 \alpha^1_{31} + p_2 p_3 \alpha^2_{23} + p_3 p_2 \alpha^3_{32}) + \notag \\
 p_2^2(p_1 p_2 \alpha^2_{12} + p_2 p_1 \alpha^2_{21} + p_1 p_3 \alpha^1_{13} + p_3 p_1 \alpha^3_{31} + p_2 p_3 \alpha^2_{23} + p_3 p_2 \alpha^2_{32}) + \notag \\
p_3^2(p_1 p_2 \alpha^1_{12} + p_2 p_1 \alpha^2_{21} + p_1 p_3 \alpha^3_{13} + p_3 p_1 \alpha^3_{31} + p_2 p_3 \alpha^3_{23} + p_3 p_2 \alpha^3_{32}) + \notag \\
\frac{p_1^2}{1-p_1}.  \label{eq:threeclumexp}
\end{align}
The expressions for $\mathbb{E}[M_2]$ and $\mathbb{E}[M_3]$ can be similarly computed as above. 
We can also simply interchange $p_1$ with $p_2$ in the above to compute $\mathbb{E}[M_2]$, and interchange $p_1$ with $p_3$ in the above to compute $\mathbb{E}[M_3]$. 
\end{proof}

\begin{figure} 
\begin{subfigure}[t]{0.3\textwidth}
\centering
\includegraphics[width=\textwidth]{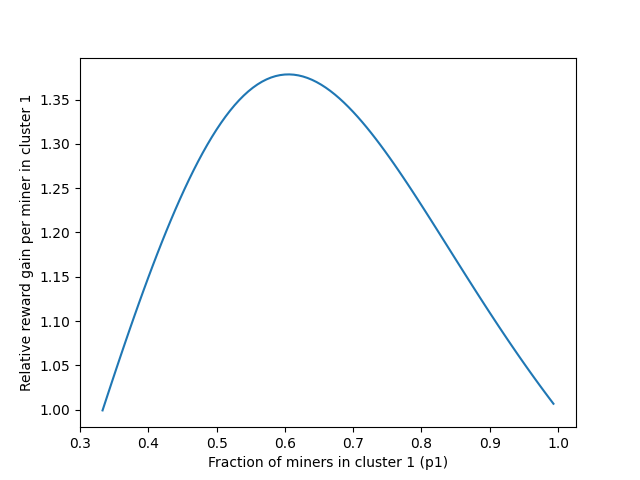} 
\caption{}
\label{fig:threeclustergain}
\end{subfigure}
\hfill 
\begin{subfigure}[t]{0.3\textwidth}
\centering
\includegraphics[width=\textwidth]{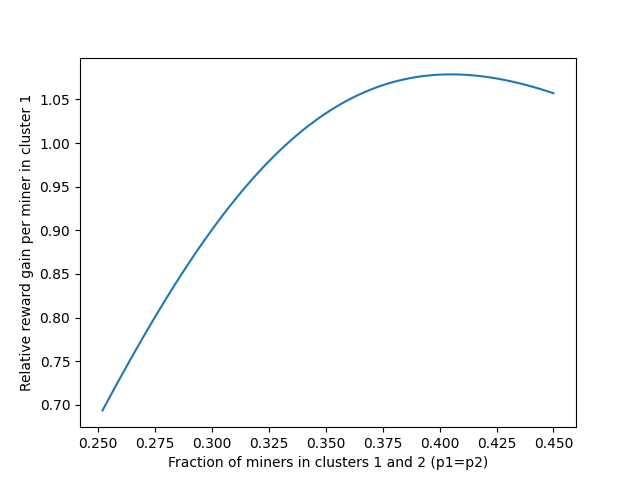} 
\caption{}
\label{fig:threeclustergain_twodominantclusters}
\end{subfigure}
\hfill 
\begin{subfigure}[t]{0.3\textwidth}
\centering
\includegraphics[width=\textwidth]{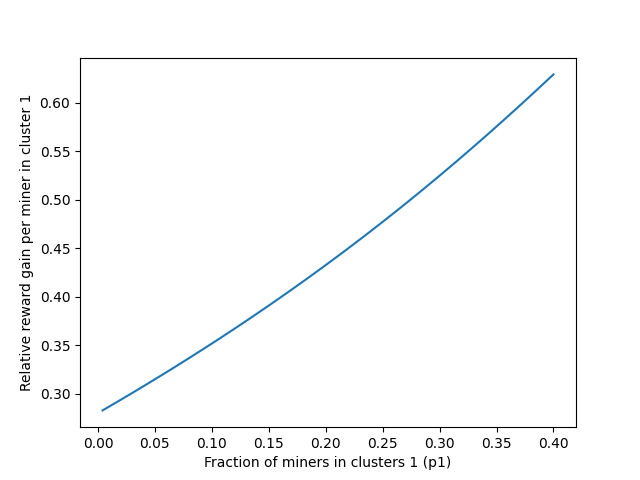} 
\caption{}
\label{fig:threeclustergain_twosmallclusters}
\end{subfigure}
\caption{(a) Relative gain of a miner $v$ in the dominant cluster, when the other two clusters  are of equal size. (b) Relative gain of a miner $v$ in cluster 1, when clusters 1 and 2 are of the same size. (c) Relative gain of a miner $v$ in cluster 1, when cluster 3's size is fixed at 60\% of the network size. }
\label{fig:threeclusterplots} 
\end{figure}

%
\smallskip
\noindent
{\bf Implication.} 
Figure~\ref{fig:threeclusterplots} presents how $\mathbb{E}[F_v]$ varies under different values for parameters $p_1, p_2, p_3$ for a miner $v$ in the first cluster. 
First, we consider the case where cluster 1 is the unique dominant cluster, while cluster 2 and 3 are equal sized smaller clusters. 
Figure~\ref{fig:threeclustergain} plots the percentage of blocks mined by a miner in cluster 1 as its size varies between 33\% to 100\% of the network size. 
The optimum gain for miner $v$ occurs when the size of its cluster is around 60\% of the network size. 

Next, we consider a setting where there are two dominant clusters, cluster 1 and 2, both of the same size. 
Figure~\ref{fig:threeclustergain_twodominantclusters} plots the rewards for miner $v$ in this scenario when the dominant cluster size varies between 25\% to 45\% of network size. 
We see that the relative gain increases beyond 1 only when the dominant cluster size is greater than 33\% of the network size. 

In the third plot, we consider a scenario where there is a single dominant cluster of size equal to 60\% of entire network. 
We then vary the size of cluster 1 between 0 and 40\% as shown in Figure~\ref{fig:threeclustergain_twosmallclusters}. 
We observe that the larger the size of cluster 1, the better it is for miner $v$. 
But in all cases, the relative gain stays below 1.

\label{s:threeclusters}

\section{\S\ref{s: evaluation} Simulation Setup Details}
\label{apx:simdetails}

\begin{figure}[t]
    \centering
    \includegraphics[width=0.98\textwidth]{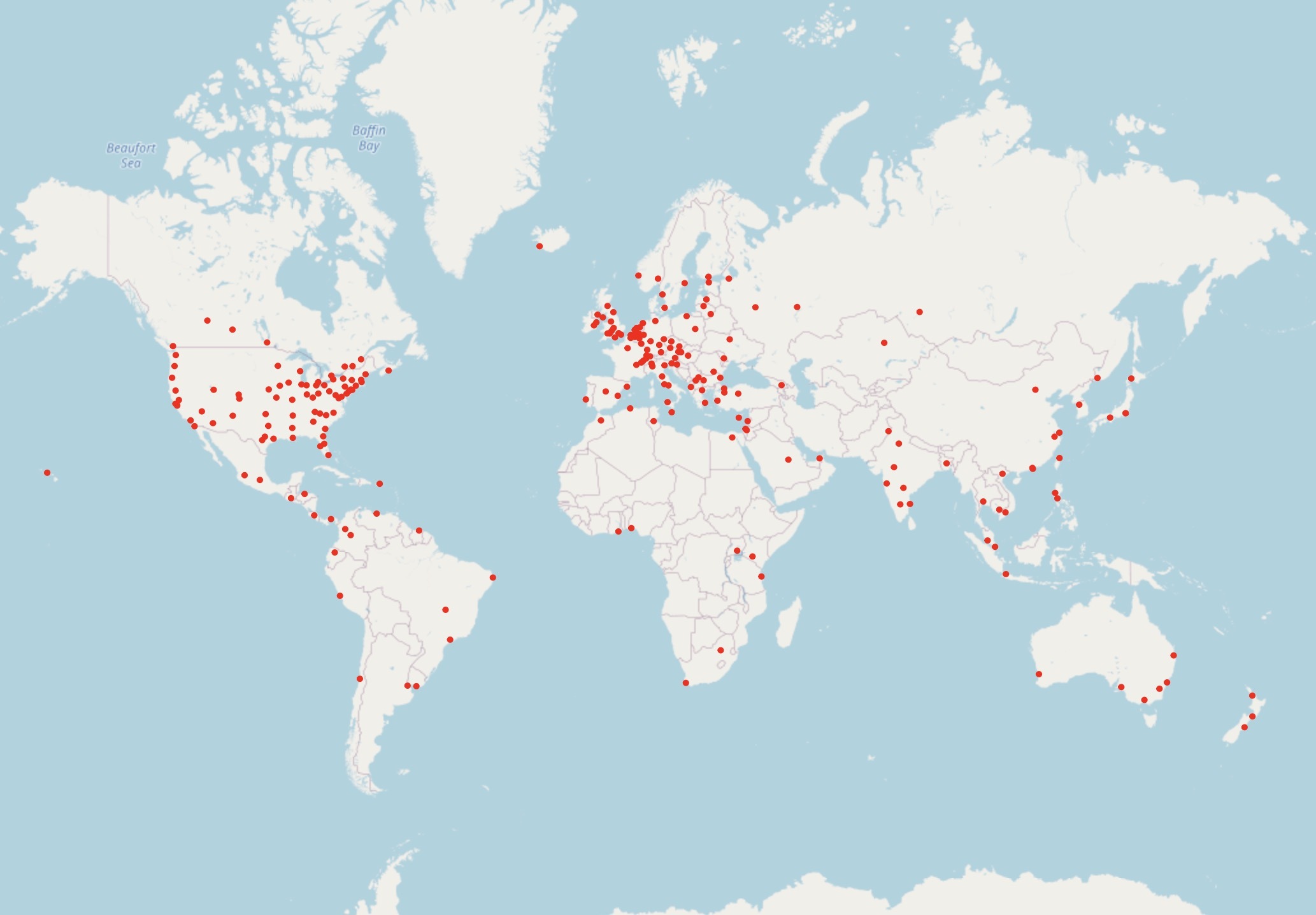}
    \caption{Map showing the locations of the 246 miners used in experiments.}
    \label{fig:map}
\end{figure}

\begin{figure}[t]
    \centering
    \includegraphics[width=0.98\textwidth]{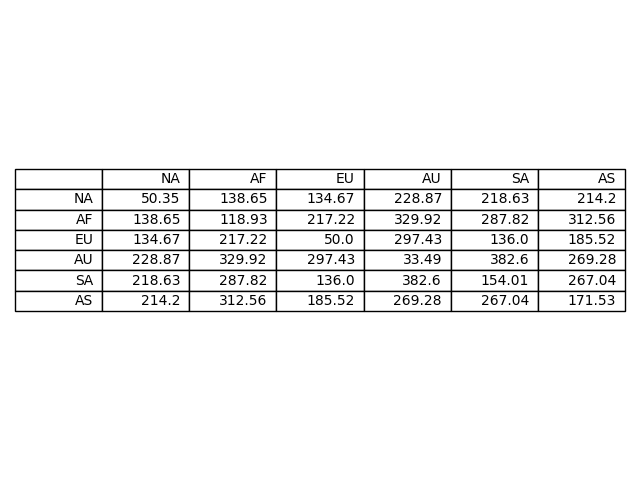}
    \caption{Average latency (ms) between miners from different continents. NA: North America, AF: Africa, EU: Europe, AU: Australia, SA: South America, AS: Asia. }
    \label{fig:continTable}
\end{figure}
 
\end{document}